\documentclass[reqno, 11pt]{amsart}%
\usepackage{amssymb}
\usepackage[nesting]{hyperref}
\usepackage[pdftex]{graphicx}
\usepackage{listings}
\usepackage{multirow}
\usepackage{placeins}
\usepackage{color}
\usepackage{subfigure}
\usepackage{lscape}
\usepackage{dsfont}
\usepackage{amsmath}
\usepackage{amsfonts}%
\usepackage{tikz}
\usepackage{verbatim}
\providecommand{\U}[1]{\protect\rule{.1in}{.1in}}
\newcommand{\BlackBoxes}{\global\overfullrule5pt}

\BlackBoxes

\newcommand{\R}{\mathbb{R}}
\newcommand{\N}{\mathbb{N}}

\newcommand{\Eop}{\mathbb{E}}
\newcommand{\Pop}{\mathbb{P}}
\newcommand{\Q}{\mathbb{Q}}

\newcommand{\seq}[1]{(#1_n)_{n\in\N}}

\newcommand{\indicator}{{\mathds 1}}
\newcommand{\Ind}[1]{\indicator_{\{#1\}}}

\newcommand{\prT}[1]{#1=(#1_t)_{t\in[0,T]}}

\newcommand{\xib}{(\xi,b)}

\newcommand{\stprT}[1]{(#1_t)_{t\in[0,T]}}
\newcommand{\stT}{{s\in[t,T]}}

\newtheorem{theorem}{Theorem}
\newtheorem{corollary}[theorem]{Corollary}
\newtheorem{lemma}[theorem]{Lemma}
\newtheorem{proposition}[theorem]{Proposition}
\theoremstyle{definition}

\newtheorem{definition}[theorem]{Definition}
\numberwithin{equation}{section} \numberwithin{theorem}{section}
\def\0{\kern0pt\-\nobreak\hskip0pt\relax}

\makeatletter
\AtBeginDocument{ \def\@serieslogo{ \vbox to\headheight{ \parindent\z@ \fontsize{6}{7\p@}\selectfont
\vss}}}

\def\makeoverbar#1#2#3#4#5#6#7{ \setbox0=\hbox{$\m@th#2\mkern#5mu{{}#3{}}\mkern#6mu$} \setbox1=\null \dimen@=#4\fontdimen8#13 \dimen@=3.5\dimen@
\advance\dimen@ by \ht0 \dimen@=-#7\dimen@ \advance\dimen@ by \wd0
\ht1=\ht0 \dp1=\dp0 \wd1=\dimen@
\dimen@=\fontdimen8#13 \fontdimen8#13=#4\fontdimen8#13
\rlap{\hbox to \wd0{$\m@th\hss#2{\overline{\box1}}\mkern#5mu$}}
\fontdimen8#13=\dimen@}
\def\mylabel#1#2{{\def\@currentlabel{#2}\label{#1}}}
\makeatother

\begin{document}
\title[Optimal investment and reinsurance problems]{Bayesian optimal investment and reinsurance with dependent financial and insurance risks}
\author[N. \smash{B\"auerle}]{Nicole B\"auerle${}^*$}
\address[N. B\"auerle]{Department of Mathematics,
Karlsruhe Institute of Technology (KIT), D-76128 Karlsruhe, Germany}

\email{\href{mailto:nicole.baeuerle@kit.edu}
{nicole.baeuerle@kit.edu}}

\author[G. \smash{Leimcke}]{Gregor Leimcke${}^*$}
\address[G. Leimcke]{Department of Mathematics,
Karlsruhe Institute of Technology (KIT), D-76128 Karlsruhe, Germany}

\email{\href{gregor.leimcke@mail.de} {gregor.leimcke@mail.de}}

\thanks{${}^*$ Department of Mathematics,
Karlsruhe Institute of Technology (KIT), D-76128 Karlsruhe, Germany}
\begin{abstract}
Major events like natural catastrophes or the COVID-19 crisis have impact both on the financial market and on claim arrival intensities and claim sizes of insurers. Thus, when optimal  investment and reinsurance  strategies have to be determined  it is important to consider models which reflect this dependence. In this paper we make a proposal how to generate dependence between the financial market and claim sizes in times of crisis and determine via a stochastic control approach an optimal  investment and reinsurance  strategy which maximizes the expected exponential utility of terminal wealth. Moreover, we also allow that the claim size distribution may be learned in the model. We give comparisons and bounds on the optimal strategy using simple models. What turns out to be very surprising is that numerical results indicate that even a  minimal dependence which is created in this model has a huge impact on the control in the sense that the insurer is much more prudent then.

\end{abstract}
\maketitle


\makeatletter \providecommand\@dotsep{5} \makeatother



\vspace{0.5cm}
\begin{minipage}{14cm}
{\small
\begin{description}
\item[\rm \textsc{ Key words} ]
{\small Risk Theory, Stochastic Control, Dependence Modeling, Learning, Bayesian Model}
\item[\rm \textsc{JEL C61, C11, G22}]
\end{description}
}
\end{minipage}

\section{Introduction}

The COVID-19 pandemic has a significant impact on individuals, society and almost all sectors of the economy. This also applies to the insurance industry as well as the financial market by drop down of asset prices. The COVID-19 crisis is one example for an event with impact on financial and insurance risks, which shows that it makes sense to add interdependencies between both. This is also suggested by Wang et al.\ \cite{Wang2018}, who point out the following two reasons: First, (re)insurance companies transfer their insurance risks to the capital market by using insurance-linked securities, like catastrophe bonds, for instance. As a result, an insurer invested in the financial market is exposed to the insurance risks exported by another insurance company to the financial market, and there may be dependencies among these risks for example through natural catastrophes.
A second interconnectedness among financial and insurance risks is in insurance contracts for financial guarantees, which can cause systemic risk. 

Whereas  it is common now in the actuarial literature to model dependencies between different lines of business, the number of papers which connect the evolution of the financial market to the occurrence of claims is sparse. 

A widespread approach to obtain dependent business lines is by {\em common shock } models. In general this means that there is an additional Poisson process which produces joint claims in all or many business lines. Papers which have used this  approach are among others \cite{BaeuerleGruebel2005,GBC12,YuenLiangZhou2015,BLX16,BiChen2019,BaeuerleLeimcke2020}. The first two papers in this list deal with modeling and computational aspects of performance measures, whereas the last four use these models to solve stochastic control problems for optimal reinsurance and investment for different criteria and for diffusion as well as jump models. The advantage of modeling dependence in this way is that we obtain an immediate interpretation for the interdependence. Since the papers \cite{YuenLiangZhou2015,BLX16,BiChen2019,BaeuerleLeimcke2020} consider a financial market which is independent from the claim generation mechanism, the control problems for investment and reinsurance decompose which makes it of course easier to obtain explicit solutions. Another popular approach to model dependence between business lines is to use {\em L\'evy copuals}, see among others \cite{bk05,bb11,atwy14} and \cite{acw11} for an overview. This approach is elegant from a mathematical point of view but its interpretation is less clear than for common shock models. Other approaches include the construction of dependence via {\em interacting intensities} (see \cite{BaeuerleGruebel2008}) or a {\em common subordinator} (see \cite{SchererSelch2018}).

The first contribution of this paper is to model a dependence between the financial market and the insurance business for the joint problem of optimal investment and (proportional) reinsurance. To keep the model simple we restrict here to one business line for the insurance risk, but the model can be extended here in a straightforward way. A paper which connects financial and insurance risk is  \cite{Wang2018} where a discrete-time risk model is considered. The authors there assume a joint distribution for the claim size and the discount factor at each point in time and are interested in the asymptotics for the finite-time ruin. They do not consider a control. The second paper is the recent paper by \cite{bs20} who create the dependence by a common factor process which influences drift and volatility of the risky asset as well as size and risk fluctuations of the insurance risk process. They consider a diffusion model and general utility function and obtain explicit solutions in some special cases. In contrast to their approach we assume here that in 'normal' times we have independence and that  dependence is created by major events like catastrophes. More precisely, whenever the claim size exceeds a certain threshold we assume that this corresponds to a catastrophe and implies at the same time a drop of the risky asset by a random proportion. What turns out to be very surprising is the fact that creating only a small dependence has a sincere effect on the optimal investment  strategy.

Our second contribution is that we allow the claim size distribution to be learned. In most articles, it is assumed  that the insurer has complete knowledge of the model. However, in reality, insurance companies operate in a setting with partial information. That is, with regard to the net claim process, only the claim arrival times and magnitudes are directly observable. Therefore we study the optimal investment and reinsurance problem in a partial information framework. More precisely we consider a Bayesian approach and restrict here to the claim size distribution which is allowed to be learned from a finite set of possible distributions (for learning the intensity see e.g. \cite{BaeuerleLeimcke2020}). 
A paper with learning in an actuarial context is \cite{LST14} where dividend payment is optimized and the drift of the risky asset has to be learned. The model there is a diffusion model. \cite{s15,LiangBayraktar2014} are both Hidden-Markov models which means that a latent hidden factor influences model parameters. In \cite{s15} again the dividend has to be maximized in a diffusion setting with unobservable dirft.   Based on the suggestion in \cite[p.\,165]{AsmussenAlbrecher2010}, the authors in \cite{LiangBayraktar2014} consider the optimal investment and reinsurance problem for maximizing exponential utility under the assumption that the claim intensity and loss distribution depend on the states of the Hidden Markov chain.

The aim in our paper is to maximize the expected  exponential utility  of the insurer's capital at a fixed time point. Note that this is an interesting optimization criterion which interpolates between a mean-variance criterion and a robust approach (for details see \cite{BaeuerleLeimcke2020}). The control consists of (proportional) reinsurance and investment into two assets. The baseline financial market is given by a Black Scholes model and the insurance model is a Cram\'er-Lundberg model. As explained before, as soon as the claim size exceeds a threshold the risky asset drops by a random proportion. Using stochastic control methods we are able to characterize the optimal investment and reinsurance strategy via the Hamilton-Jacobi-Bellman (HJB) equation. 
Since the value function may not be differentiable everywhere we use the Clarke gradient as a general gradient in our analysis. In the case of known model data we get explicit optimal investment and reinsurance strategies and are able to discuss the influence of the threshold level which creates the dependency.

The paper is organized as follows: In the next section we introduce our basic model which consists of the claim arrival process, the financial market, the strategies and the optimization problem. In Section \ref{sec:learn} we state the model with learning and explain how we can transform the model with unknown claim size distribution to a model with known data. The standard approach here is to include a filter process which keeps track of all relevant observations. Section \ref{sec:sol} contains the solution. Being able to show that the value function possesses some Lipschitz properties we can prove that it is a solution of a generalized HJB equation where we replace a derivative by the generalized Clarke gradient. Thus, we are also able to characterize an optimal pair of investment and reinsurance strategy.
Due to the dependence between the financial market and the claim  process these strategies are now rather complicated. So we first manage in Section \ref{sec:comp} to compare the optimal strategy to the optimal one where we have independence between the financial market and claim occurrence. It will turn out that the insurance company will invest less  when dependence shows up. Indeed a numerical example will reveal the magnitude of the impact of the threshold which creates the dependence. We can show that even large thresholds which create a minimal dependence have a huge impact on the investment strategy. Second we are able to compare the optimal investment strategy in our model to the optimal one in a model with known data and where the jump size distribution exactly equals  our expectation. We will see that in the latter model the invested amount provides an upper bound to what is invested in the more complicated model. In the appendix we summarize additional information on the Clarke gradient and provide detailed calculations and proofs for our main theorems.

\section{The Optimal Investment and Reinsurance Model}\label{sec:model}
We consider an insurance company with the aim to maximize the expected utility of the terminal surplus by choosing optimal investment and reinsurance strategies. The processes $\Psi$ and $W$ below are defined on a common probability space $(\Omega,\mathcal{F},\Pop)$.

\subsection{The aggregated claim amount process}
In the following, let $N=(N_t)_{t\ge0}$ be a Poisson process with intensity $\lambda>0$.  We interpret the jump times of $N$, denoted by $(T_n)_{n\in\N}$, as arrival times of insurance claims. We assume that $(Y_n)_{n\in\N}$ is a sequence of positive random variables, where $Y_n$ describes the claim size at $T_n$.
The insurer faces uncertainty about the claim size distribution. This is taken into account by a Bayesian approach.
Let $\{F_\vartheta:\vartheta\in\Theta\}$, $\Theta\subset\R^n$, be a family of distributions on $(0,\infty)$, where $\vartheta$ in unknown. 
We view $\vartheta$ as a random variable taking values in $\Theta=\{1,\ldots,m\}$ for some $m\in\N$ and initial distribution $\pi_\vartheta(j)$, $j=1,\ldots,m$.
Moreover, we suppose that $F_j$ is absolutely continuous with density $f_j$, where
\begin{equation*}
  M_{j}(z):= \int_{(0,\infty)} e^{zy}f_j(y)dy<\infty,\quad z\in\R,\quad j=1,\ldots,m.
\end{equation*}
The sequence $Y_1,Y_2,\ldots$ is assumed to be conditional independent and identically distributed according to $F_\vartheta$ given $\vartheta$ as well as independent of $(T_n)_{n\in\N}$.
The aggregated claim amount process, denoted by $(S_t)_{t\ge0}$, is given by
\begin{equation*}
S_t = \sum_{i=1}^{N_t} Y_i = \int_0^t y\, \Phi(dt,dy),
\end{equation*}
where $\Phi:=(T_n,Y_n)$ is the $(0,\infty)$-Marked Point Process which carries the information about the claim arrival time and amounts.

\subsection{The financial market} 
The surplus will be invested by the insurer into a financial market, where it is assumed that there exists one risk-free asset and one risky asset. The price process of the \emph{risk-free asset}, denoted by $B=(B_t)_{t\ge0}$, is given by
\begin{equation*}
d B_t = rB_t dt, \quad B_0=1,
\end{equation*}
with \emph{risk-free interest rate}  $r\in\R$. That is, $B_t = e^{rt}$ for all $t\ge0$.
The price of the risky asset drops down by a random value at the claim arrival time $T_n$, if the corresponding insurance claim $Y_n$ exceed a fixed threshold $L>1$. 
We assume that $(Z_n)_{n\in\N}$ is a sequence of independent and identically distributed random variables taking values in $(0,1)$ with distribution $Q$. 
It is supposed that $(Z_n)_{n\in\N}$ is independent of $(T_n)_{n\in\N}$ and $(Y_n)_{n\in\N}$.
The random variable $Z_n$ describes the relative jump height downwards of the risky asset at time $T_n$, if $Y_n>L$. From now on, we set $\Psi := (T_n,(Y_n,Z_n))_{n\in\N}$ and let $E:=(0,\infty)\times(0,1)$. That is, $\Psi$ is the $E$-Marked Point Process which contains the information of the claim arrival times, claim sizes and potential relative jumps downwards of the risky asset. The filtration generated by $\Psi$ is denoted by $\mathfrak{F}^\Psi=(\mathcal{F}_t^\Psi)_{t\ge0}$.
The price process of the risky asset evolves according to a geometric Brownian motion between the jumps. That is, the price process of the \emph{risky asset}, denoted by  $P=(P_t)_{t\ge0}$, is characterized by
\begin{equation*}
d P_t = P_{t-}\bigg(\mu dt + \sigma d W_t - \int_E z \mathds{1}_{(L,\infty)}(y)\Psi(dt,d(y,z))\bigg) , \quad P_0=1,
\end{equation*}
where $\mu\in\R$ and $\sigma>0$ are constants describing the drift and volatility of the risky asset, respectively, and $(W_t)_{t\ge0}$ is a standard Brownian motion which is independent of $\seq{T}$, $\seq{Y}$ and $\seq{Z}$. Since the price process of the risky asset is observable, the filtration generated by $P$, denoted by $\mathfrak{F}^P=(\mathcal{F}_t^P)_{t\ge0}$, is known by the insurer. 
Throughout this work, ${\mathfrak G}=(\mathcal{G}_t)_{t\ge0}$ denotes the observable filtration of the insurer which is given by
\begin{equation*}
{\mathcal G}_t = {\mathcal F}^P_t\vee{\mathcal F}_t^\Psi,\quad t\ge0.
\end{equation*}

\subsection{The strategies}
We assume that the wealth of the insurance company is invested into the previously described financial market. 
\begin{definition}\label{def:investment}
	An \emph{investment strategy}, denoted by $\xi=(\xi_t)_{t\ge0}$, is an $\R$-valued, c\`{a}dl\`{a}g and ${\mathfrak G}$-predictable process such that $| \xi_t|\le K$ for some $0<K<\infty$. $\xi_t$ is the amount of money invested at time $t$.
\end{definition}

The restriction  $| \xi_t|\le K$ is only a technical tool. We will make $K$ sufficiently large later, s.t.\ the optimal $\xi_t^\star$ is the same as in the unrestricted problem.

We further assume that the first-line insurer has the possibility to take a proportional reinsurance. Therefore, the \emph{part of the insurance claims paid by the insurer}, denoted by $h(b,y)$, satisfies
\begin{equation*}
h(b,y) = b\cdot y
\end{equation*}
with \emph{retention level} $b\in[0,1]$ and \emph{insurance claim} $y\in(0,\infty)$. Here we suppose that the insurer is allowed to reinsure a fraction of her/his claims with retention level $b_t\in[0,1]$ at every time $t$. 

\begin{definition}\label{def:reinsurance}
	A \emph{reinsurance strategy}, denoted by $b=(b_t)_{t\ge0}$, is a $[0,1]$-valued, c\`{a}dl\`{a}g and ${\mathfrak G}$-predictable process. 
\end{definition}
We denote by ${\mathcal U}[t,T]$ the set of all admissible strategies $(\xi,b)$ on $[t,T]$. 
We assume that the policyholder's payments to the insurance company are modelled by a fixed \emph{premium (income) rate} $c=(1+\eta)\kappa$ with safety loading $\eta>0$ and fixed constant $\kappa>0$, which means that premiums are calculated by the expected value principle.
If the insurer chooses retention levels less than one, then the insurer has to pay premiums to the reinsurer. The \emph{part of the premium rate left to the insurance company} at retention level $b\in[0,1]$, denoted by $c(b)$, is $c(b) = c - \delta(b)$, where $\delta(b)$ denotes the \emph{reinsurance premium rate}. We say $c(b)$ is the \emph{net income rate}.
Moreover, the net income rate $c(b)$ should increase in $b$, which is fulfilled by setting $\delta(b) :=  (1-b)(1+\theta)\kappa$ with $\theta>\eta$ which represents the safety loading of the reinsurer.
Therefore
\begin{equation}\label{eq:cb}
c(b) = (1+\eta)\kappa - (1-b)(1+\theta)\kappa = (\eta-\theta)\kappa + (1+\theta)\kappa\,b.
\end{equation}
This reinsurance premium model is used e.g.\ in \cite{ZhuShi2019}. 
The surplus process $(X^{\xi,b}_t)_{t\ge0}$ under an admissible investment-reinsurance strategy $\xib\in{\mathcal U}[0,T]$ is given by
\begin{align*}
d X^{\xi,b}_t
&=(X^{\xi,b}_t - \xi_t)r dt + \xi_t\bigg(\mu dt+\sigma dW_t-\int_E z\mathds{1}_{(L,\infty)}(y)\Psi(dt,d(y,z))\bigg) + c(b_t)dt - b_t dS_t  \\
&= \big(rX^{\xi,b}_t + (\mu - r)\xi_t + c(b_t)\big) dt + \xi_t\sigma dW_t - \int_E\big(b_{t}y+ \xi_t z \mathds{1}_{(L,\infty)}(y)\big)\Psi(dt,d(y,z)).
\end{align*}
We suppose that $X^{\xi,b}_0=x_0 >0$  is the initial capital of the insurance company.

\subsection{The optimization problem}
Clearly, the insurance company is interested in an optimal investment-reinsurance strategy. But there are various optimality criteria to specify optimization of proportional reinsurance and investment strategies.  We consider the expected utility of wealth at the terminal time $T>0$ as criterion with
exponential utility function $U:\R\to\R$ 
\begin{equation}\label{eq:u}
U(x)=-e^{-\alpha x},
\end{equation}
where the parameter $\alpha>0$ measures the \emph{degree of risk aversion}. The exponential utility function is useful since by choosing $\alpha$ we can interpolate between a risk-sensitive criterion and a robust point of view as explained in \cite{BaeuerleLeimcke2020}. The case of small $\alpha $ can be seen as maximizing the expectation with a bound on the variance and the case of large $\alpha$ can be seen as a robust optimization.

Next, we are going to formulate the dynamic optimization problem. 
We define the value functions, for any $(t,x)\in[0,T]\times\R$ and $(\xi,b)\in{\mathcal U}[t,T]$, by
\begin{equation}\label{eq:problem} 
\begin{aligned}
V^{\xi,b}(t,x) &:= \Eop^{t,x}\big[U(X^{\xi,b}_T)\big],  \\
V(t,x) &:= \sup_{(\xi,b)\in{\mathcal U}[t,T]}V^{\xi,b}(t,x).
\end{aligned}
\end{equation}
The expectation $\Eop$ is taken w.r.t.\ the  probability measure $\pi_\vartheta \otimes\Pop$ and $\Eop^{t,x}$ denotes the conditional expectation given $X^{\xi,b}_t=x$.

\section{A Model with Learning}\label{sec:learn}

The task is to reduce the control problem~\eqref{eq:problem} with partial information within the introduced framework to one with complete information, taken the observations into account. 

\subsection{Filtering}
By the Bayes rule, the posterior probability mass function of $\vartheta$ given the observation $\bar{Y}_n=\bar{y}_n$ with $\bar{Y}_n:=(Y_1,\ldots,Y_n)$ and $\bar{y}_n:= (y_1,\ldots,y_n)$ is 
\begin{equation}\label{posttheta}
\Pop(\vartheta=j|\bar{Y}_n=\bar{y}_n) 
= \frac{\pi_\vartheta(j)\prod_{i=1}^n f_j(y_i)}{\sum_{k=1}^m\pi_\vartheta(k)\prod_{i=1}^n f_k(y_i)},\quad j=1,\ldots,m.
\end{equation}
However, the solution method requires a dynamic representation of this posterior probability distribution given the information up to any time $t$. To achieve this, let us introduce the following notation.
Throughout this paper, we write
\begin{equation*}
  p_j(t) = \Pop(\vartheta=j|\mathcal{F}_t^\Psi),\quad t\ge0,\quad j=1,\ldots,m.
\end{equation*}
Moreover, let $(p_t)_{t\ge0}$ denote the $m$-dimensional process defined by 
\begin{equation*}
  p_t:=(p_1(t),\ldots,p_m(t)),\quad t\ge0.
\end{equation*}
We obtain the following representation of the process $(p_t)_{t\ge0}$ from \eqref{posttheta}:
\begin{equation}\label{pj}
p_j(t) = \pi_\vartheta(j) + \int_0^t\int_{(0,\infty)}\bigg(\frac{p_j(s-)\,f_j(y)}{\sum_{k=1}^m p_k(s-)\,f_k(y)}-p_j(s-)\bigg)\Phi(ds,dy),\quad j=1,\ldots,m.
\end{equation}
Note that $(p_t)_{t\ge0}$ is a pure jump process and the new state of $(p_t)$ at the jump time $T_n$ with jump sizes $Y_n$ is 
\begin{equation*}
p_{T_n} = J\big(p_{T_n-},Y_n\big),\quad n\in\N,
\end{equation*}
where
\begin{equation*}
J(p,y) := \left(\frac{f_1(y)\, p_1}{\sum_{k=1}^m f_k(y)\,p_k},\ldots,\frac{f_m(y)\,p_m}{\sum_{k=1}^m f_k(y)\,p_k}\right),
\end{equation*}
for $p=(p_1,\ldots,p_m)\in\Delta_m:=\{x\in\R_+^m:\sum_{k=1}^m x_i=1\}$ and $y\in(0,\infty)$.

\begin{proposition}\label{GintkernelPsi}
	The ${\mathfrak G}$-intensity kernel of $\Psi=(T_n,(Y_n,Z_n))$, denoted by  $\hat{\nu}(t,d(y,z))$, is given by
	\begin{equation*}
	\hat{\nu}(t,d(y,z)) = \lambda\sum_{k=1}^m p_k(t)f_k(y)dyQ(dz),\quad t\ge0.
	\end{equation*}
\end{proposition}

\begin{proof}
First note that $\hat{\nu}$ is a transition kernel. The  ${\mathfrak G}$-intensity is derived from the ${\mathfrak G} \vee \sigma(\vartheta)$-intensity kernel $\lambda f_\vartheta(y)dyQ(dz)$ by conditioning on ${\mathcal G}_t$ (see \cite{bre}). Note here in particular that the posterior predictive distribution of the claim sizes given the observed claims up to time $t$ is $\sum_{k=1}^m p_k(t)f_k(y)dy$.
\end{proof}

We denote by $\hat{\Psi}(dt, d(y,z))$ the compensated random measure given by
\begin{equation}\label{Psihat}
\hat{\Psi}(dt, d(y,z)) :=  \Psi(dt, d(y,z)) - \hat\nu(t,d(y,z))dt,
\end{equation}
where $\hat\nu$ is defined as in Proposition~\ref{GintkernelPsi}.
Thus, we obtain the following indistinguishable representation of the surplus process $(X^{\xi,b}_t)_{t\ge0}$:
\begin{equation}\label{wealth}
\begin{aligned}
d X^{\xi,b}_t = \bigg(&r X_t^{\xi,b} + (\mu-r)\xi_t + c(b_t) - \lambda\sum_{k=1}^m p_k(t)\big(b_t\mu_k + \xi_t \bar F_k(L)\Eop[Z]\big)\bigg)dt  \\
& + \xi_t\sigma dW_t - \int_E \big(b_t y + \xi_t z \mathds{1}_{(L,\infty)}(y)\big) \hat{\Psi}(dt, d(y,z)),\quad t\ge0,
\end{aligned}
\end{equation}
where $\mu_j:=\int_{(0,\infty)} y f_j(y)dy$, $\bar F_j$ denotes the survival function of $F_j$, $j=1,\ldots,m$, and $Z$ is a random variable with $Z\sim Z_1$. Note that all processes here are ${\mathfrak G}$-adapted.
This dynamic will be one part of the reduced control model discussed in the next section. 

\subsection{The Reduced Control Problem}
The process $(p_t)_{t\ge0}$ in \eqref{pj} carries all relevant information about the unknown parameter $\vartheta$ contained in the observable filtration ${\mathfrak G}$ of the insurer.  Therefore, the state process of the reduced control problem with complete observation is the $(m+1)$-dimensional process
\begin{equation*}
(X^{\xi,b}_s,p_s)_{s\in[t,T]},
\end{equation*}
where $(X^{\xi,b}_s)$ is given by \eqref{wealth} and $(p_s)$ is given by \eqref{pj} for some fixed initial time $t\in[0,T)$ and $\xib\in{\mathcal U}[t,T]$. 
We can now formulate the reduced control problem. For any $(t,x,p)\in[0,T]\times\R\times\Delta_m$, the value functions are given by
\begin{equation}\label{P} \tag{P}
\begin{aligned}
V^{\xi,b}(t,x,p) &:= \Eop^{t,x,p}\big[U(X^{\xi,b}_T)\big],\\
V(t,x,p) &:= \sup_{(\xi,b)\in{\mathcal U}[t,T]}V^{\xi,b}(t,x,p),
\end{aligned}
\end{equation}
where $ \Eop^{t,x,p}$ denotes the conditional expectation given $(X^{\xi,b}_t,p_t)=(x,p)$. An investment-reinsurance strategy $(\xi^\star,b^\star)\in{\mathcal U}[t,T]$ is optimal if
$V(t,x,p) = V^{\xi^\star,b^\star}(t,x,p).$
Note that by classical filtering results we have that $V(0,x,\pi_\vartheta)=V(0,x)$
(see e.g. \cite{BaeuerleRieder2007}).


\section{The Solution}\label{sec:sol}

\subsection{The HJB equation}\label{sec:HJB}

In a first step we derive the HJB equation for the value function $V$ using standard methods and assuming full differentiability of $V$, which results in
\begin{equation}\label{HJBV}
\begin{aligned}
&0=\sup_{(\xi,b)\in[-K,K]\times[0,1]} \bigg\{V_t(t,x,p) - \lambda V(t,x,p) +V_x(t,x,p)\big(rx + (\mu-r)\xi+c(b)\big)  \\
&+ \frac12\sigma^2V_{xx}(t,x,p)\xi^2 + \lambda \sum_{k=1}^m p_k\int_E V\big(t,x-(b y + z\xi\mathds{1}_{(L,\infty)}(y)),J(p,y)\big)f_k(y)dyQ(dz)\bigg\},
\end{aligned}
\end{equation}
For solving \eqref{HJBV} we apply the usual separation approach: For any $(t,x,p)\in[0,T]\times\R\times\Delta_m$, we assume
\begin{equation}\label{separation}
V(t,x,p) = -e^{-\alpha x e^{r(T-t)}}g(t,p)
\end{equation}
with $g\ge 0$. This implies that we conclude from~\eqref{HJBV}
\begin{equation}\label{HJBgdiff}
\begin{aligned}
0&=\inf_{(\xi,b)\in[-K,K]\times[0,1]} \bigg\{g_t(t,p) - \lambda g(t,p) - \alpha e^{r(T-t)}g(t,p)\Big((\mu-r)\xi + c(b) - \frac12\alpha \sigma^2 e^{r(T-t)}\xi^2\Big)  \\
&\quad+\lambda \sum_{k=1}^m p_k \int_0^\infty g(t,J(p,y))e^{\alpha b y e^{r(T-t)}} \int_{(0,1)} e^{\alpha \xi z \mathds{1}_{(L,\infty)}(y)e^{r(T-t)}}Q(dz) f_k(y)dy \bigg\}.
\end{aligned}
\end{equation}
However, $V$ is probably not differentiable w.r.t.\ $t$. Assuming $t\mapsto g(t,p)$ is Lipschitz on $[0,T]$ for all $p\in\Delta_m$, we can replace the partial derivative of $g$ w.r.t.\ $t$ by Clarke's generalized subdifferential (see appendix).
 Throughout, we denote by ${\mathcal L}$ an operator acting on functions $g:[0,T]\times\Delta_m\to(0,\infty)$ and $(\xi,b)\in [-K,K]\times[0,1]$ which is defined by
\begin{equation}\label{L}
{\mathcal L} g(t,p;\xi,b) :=  - \lambda g(t,p)+ \alpha e^{r(T-t)}g(t,p)(\theta-\eta)\kappa + \gamma(t,p,\xi,b),
\end{equation}
where
\begin{equation}\label{eq:gamma}
\begin{aligned}
 \gamma(t,p,\xi,b) := &-\alpha e^{r(T-t)} g(t,p)\Big((\mu-r)\xi -\frac12\alpha\sigma^2 e^{r(T-t)}\xi^2 + (1+\theta) \kappa b\Big) \\
&+\lambda\sum_{k=1}^m p_k \int_0^\infty g(t,J(p,y)) e^{\alpha b y e^{r(T-t)}}\int_{(0,1)}e^{\alpha \xi z \mathds{1}_{(L,\infty)}(y) e^{r(T-t)}}Q(dz)f_k(y)dy.
\end{aligned}
\end{equation}
Using this operator and replacing the partial derivative of $g$ w.r.t.\ $t$,  in~\eqref{HJBgdiff} by Clarke's generalized subdifferential, we get the generalized HJB equation for $g$:
\begin{equation}\label{HJBg}
0 = \inf_{(\xi,b)\in[-K,K]\times[0,1]}\big\{ {\mathcal L} g(t,p;\xi,b)\big\} + \inf_{\varphi\in\partial^C\! g_p(t)}\{\varphi\}
\end{equation}
for all $(t,p)\in[0,T)\times\Delta_m$ with boundary condition
\begin{equation}\label{HJBgbound}
g(T,p) = 1,\quad p\in\Delta_m.
\end{equation}
Note that  we set $\partial^C\! g_p(t)=\{g_p^\prime(t)\}$ at the points $t$ where the subdifferential exists.
The notation $g_p(t)$ indicates that the derivative is w.r.t.\ $t$ for fixed $p$.

\subsection{Candidate for an optimal strategy}
To obtain candidates for an optimal strategy, we have to minimize the function $\gamma$ given in \eqref{eq:gamma} w.r.t.\ $(\xi,b)$ for fixed $(t,p)$.
For this purpose we introduce the following notation:
\begin{equation*}
  M_Z(u) := \Eop\big[e^{uZ}\big], \quad u\in\R.
\end{equation*}
Notice that $M_Z^\prime(u)=\Eop\big[Z e^{uZ}\big]$ and $M_Z^{\prime\prime}(u)=\Eop\big[Z^2e^{uZ}\big]$ whenever they exist.

\begin{lemma}\label{gamma}
	For any $(t,p)\in[0,T]\times\Delta_m$, the function $\R^2\ni (\xi,b)\mapsto \gamma(t,p,\xi,b)$ is strictly convex and
	\begin{align*}
	\frac{\partial}{\partial \xi}\gamma(t,p,\xi,b) &= -\alpha e^{r(T-t)}g(t,p)\big((\mu-r)-\alpha\sigma^2 e^{r(T-t)}\xi\big) \\
   &\quad + \lambda\,\alpha\,e^{r(T-t)}\sum_{k=1}^m p_k \int_L^\infty g(t,J(p,y))e^{\alpha b y e^{r(T-t)}}f_k(y)dy\,M_Z^\prime(\alpha\,e^{r(T-t)}\xi),\\
   \frac{\partial}{\partial b}\gamma(t,p,\xi,b) &= -\alpha\,e^{r(T-t)}g(t,p)\,(1+\theta)\kappa \\
   &\quad + \lambda\alpha e^{r(T-t)}\sum_{k=1}^m p_k\!\!\int_0^\infty\!\!\! yg(t,J(p,y))e^{\alpha b y e^{r(T-t)}}\!\!\int_{(0,1)}\!\!\!e^{\alpha \xi z\mathds{1}_{(L,\infty)}(y) e^{r(T-t)}}Q(dz)f_k(y)dy.
	\end{align*}
\end{lemma}

\begin{proof}
   A straightforward calculation yields the announced partial derivatives and
	\begin{align*}
   \frac{\partial^2\gamma(t,p,\xi,b)}{\partial \xi^2} &= \alpha^2\sigma^2e^{2r(T-t)}g(t,p) \\
   &\quad + \lambda\alpha^2e^{2r(T-t)}\sum_{k=1}^m p_k\int_L^\infty\! g(t,J(p,y))e^{\alpha b y e^{r(T-t)}}f_k(y)dy\,M_Z^{\prime\prime}(\alpha e^{r(T-t)}\xi),\\
\frac{\partial^2\gamma(t,p,\xi,b)}{\partial b^2} &= \lambda\alpha^2e^{2r(T-t)}\sum_{k=1}^m p_k\bigg(\int_0^L y^2g(t,J(p,y))e^{\alpha b y e^{r(T-t)}}f_k(y)dy\\
&\quad +\int_L^\infty y^2g(t,J(p,y))e^{\alpha b y e^{r(T-t)}}f_k(y)dy M_Z(\alpha e^{r(T-t)}\xi)\bigg), \\
\frac{\partial^2\gamma(t,p,\xi,b)}{\partial b\partial\xi} &= \lambda\alpha^2e^{2r(T-t)}\sum_{k=1}^m p_k\int_L^\infty yg(t,J(p,y))e^{\alpha b y e^{r(T-t)}}f_k(y)dy M_Z^\prime(\alpha e^{r(T-t)}\xi).
\end{align*}
Therefore, the Hessian matrix $H_\gamma$ of $\gamma$ w.r.t.\ $(\xi,b)$ is given by
\begin{equation*}
H_\gamma=\alpha^2e^{2r(T-t)}\Big(A+\lambda\sum_{k=1}^m p_k B_k\Big)
\end{equation*}
with
\begingroup
\renewcommand*{\arraystretch}{1.4}
\begin{equation*}
A:=\begin{pmatrix}
\sigma^2g(t,p) & 0 \\
0 & \lambda\sum_{k=1}^m p_k\int_0^L y^2g(t,J(p,y))e^{\alpha b y e^{r(T-t)}}f_k(y)dy
\end{pmatrix}
\end{equation*}
and
\begin{equation*}
B_k:=\begin{pmatrix}
a_k &  b_k \\
b_k & c_k
\end{pmatrix}
\end{equation*}
\endgroup
with
\begin{align*}
a_k &:= \int_L^\infty g(t,J(p,y))e^{\alpha b y e^{r(T-t)}}f_k(y)dy M_Z^{\prime\prime}(\alpha e^{r(T-t)}\xi),\\
b_k &:= \int_L^\infty yg(t,J(p,y))e^{\alpha b y e^{r(T-t)}}f_k(y)dy M_Z^\prime(\alpha e^{r(T-t)}\xi),\\
c_k &:= \int_L^\infty y^2g(t,J(p,y))e^{\alpha b y e^{r(T-t)}}f_k(y)dy M_Z(\alpha e^{r(T-t)}\xi).
\end{align*}
for $k=1,\ldots,m$. To prove the convexity of $(x,y)\mapsto \gamma(t,p,\xi,b)$, it is sufficient to show that $H_\gamma$ is positive definite. Clearly, $A$ is positive definite. Moreover, for any $k\in\{1,\ldots,m\}$ and $\bar{x}=(x_1,x_2)\in\R^2\setminus\{0\}$, it holds (since $L>1$ and $(M_Z^\prime)^2 \le M_Z^{\prime\prime} M_Z$ by the Cauchy-Schwarz inequality)
\begin{align*}
&\bar{x} B_k \bar{x}^\top = x_1^2 a_k + 2x_1x_2b_k + x_2^2c_k \\
 &\ge \int_L^\infty \!g(t,J(p,y))e^{\alpha b y e^{r(T-t)}}f_k(y)dy\Big(x_1^2 M_Z^{\prime\prime}(\alpha e^{r(T-t)}\xi) \\
 &\quad + x_2^2 M_Z(\alpha e^{r(T-t)}\xi) + 2x_1x_2 M_Z^\prime(\alpha e^{r(T-t)}\xi)\Big) \\
&\ge \int_L^\infty \!g(t,J(p,y))e^{\alpha b y e^{r(T-t)}}f_k(y)dy\bigg(x_1^2 \frac{\big(M_Z^{\prime}(\alpha e^{r(T-t)}\xi)\big)^2}{M_Z(\alpha e^{r(T-t)}\xi)} + x_2^2 M_Z(\alpha e^{r(T-t)}\xi) \\
&\quad + 2x_1x_2 M_Z^\prime(\alpha e^{r(T-t)}\xi)\bigg) \\
&= \int_L^\infty \!g(t,J(p,y))e^{\alpha b y e^{r(T-t)}}f_k(y)dy \bigg(x_1 \frac{M_Z^{\prime}(\alpha e^{r(T-t)}\xi)}{\sqrt{M_Z(\alpha e^{r(T-t)}\xi)}} + x_2\sqrt{M_Z(\alpha e^{r(T-t)}\xi)}\bigg)^2>0.
\end{align*}
Consequently, $H_\gamma$ is positive definite.
\end{proof}

Setting $\nabla\gamma$ to zero, we obtain the following first order condition for the candidate of an optimal strategy in case $g>0$:
\begin{equation}\label{foc}
\begin{aligned}
v_1(t,p,\xi,b) &=\mu-r, \\
v_2(t,p,\xi,b) &= (1+\theta)\kappa,
\end{aligned}
\end{equation}
where
\begin{align*}
v_1(t,p,\xi,b) &:= \alpha\sigma^2 e^{r(T-t)}\xi+ \lambda\sum_{k=1}^m p_k\int_L^\infty \frac{g(t,J(p,y))}{g(t,p)}e^{\alpha b y e^{r(T-t)}}f_k(y)dy\,M_Z^\prime(\alpha\,e^{r(T-t)}\xi),\\
v_2(t,p,\xi,b) &:= \lambda\sum_{k=1}^m p_k\int_0^\infty y\frac{g(t,J(p,y))}{g(t,p)}e^{\alpha b y e^{r(T-t)}}\int_{(0,1)} e^{\alpha\xi z\mathds{1}_{(L,\infty)}(y)e^{r(T-t)}}Q(dz)f_k(y)dy.
\end{align*}
The next proposition states that this system of equations is solvable.

\begin{proposition}\label{candidates}
	For any $(t,p)\in[0,T]\times\Delta_m$, \eqref{foc} has a unique root w.r.t.\ $(\xi,b)$, denoted by $r(t,p):=(r_1(t,p),r_2(t,p))$, where $r_2(t,p)$ is increasing w.r.t.\ the safety loading parameter $\theta$ of the reinsurer. Moreover,
	it holds, 
	\begin{enumerate}
		\item $r_2(t,p) \le 0$ if $(1+\theta)\kappa \le A(t,p)$,
		\item $0< r_2(t,p) <1$ if $A(t,p) < (1+\theta)\kappa < B(t,p)$,
		\item $r_2(t,p)\ge 1$ if $(1+\theta)\kappa \ge B(t,p)$,
		\item $r_1(t,p)$ is decreasing with $r_2(t,p)$,
	\end{enumerate}
with
\begin{align*}
A(t,p) &:= v_2(t,p,r_1(t,p),0), \\
B(t,p) &:= v_2(t,p,r_1(t,p),1).
\end{align*}
\end{proposition}

\begin{proof}
   Due to the strict convexity of $\gamma$ according to Lemma \ref{gamma} and
   \begin{equation*}
   \lim_{\xi\to-\infty}\gamma(t,p,\xi,b) = \lim_{\xi\to+\infty}\gamma(t,p,\xi,b) = \lim_{b\to-\infty}\gamma(t,p,\xi,b) = \lim_{b\to+\infty}\gamma(t,p,\xi,b) = \infty,
   \end{equation*}
   there exists a unique minimizer of the function $\gamma$ w.r.t.\ $(\xi,b)$ for fixed $(t,p)$, i.e.\ \eqref{foc} has a unique root denoted by $r(t,p):=(r_1(t,p),r_2(t,p))$.
   Note that $\R\ni b\mapsto v_2(t,p,\xi,b)$ is strictly increasing and thus $A(t,p)<B(t,p)$. Then statements (a), (b) and (c) follow from considering the zeros of \eqref{foc} in $\theta$ when $b=0$ and when $b=1$. For (d) note that $\xi, b\mapsto v_1(t,p,\xi,b)$ are both increasing.
\end{proof}	

The proposition above provides the candidate for an optimal investment-reinsurance strategy. Let $K$ be large s.t. $|r_1(t,p)|\le K$ for all $t\in [0,T], p\in \Delta_m$.
For any $(t,p)\in[0,T]\times\Delta_m$, we set
\begin{equation*}
b(t,p) := \begin{cases} 0, & \theta\le A(t,p)/\kappa-1,\\
1, &\theta\ge B(t,p)/\kappa-1, \\
r_2(t,p), &\text{otherwise}.
\end{cases}
\end{equation*}
Then the candidate for an optimal investment-reinsurance strategy $(\xi^\star,b^\star)=(\xi^\star_t,b^\star_t)_{t\ge[0,T]}$ is given by 
\begin{equation*}
  b^\star_t := b(t,p_{t-}) \mbox{ and } \xi^\star_t := r_1(t,p_{t-}),
\end{equation*}
the latter equation only holds if $A(t,p_{t-})<(1+\theta)\kappa < B(t,p_{t-})$. If $b_t^\star=0$ or $b_t^\star=1$, then we have to find the minimum point of $\gamma$ on $(-\infty, \infty)\times [0,1]$. In the case $b_t^\star=0$, $\xi_t^\star$ may deviate from $r_1(t,p_{t-})$.
We have to solve $v_1(t,p,\xi,0)=\mu-r$ here, which unique root w.r.t.\ $\xi$ is denoted by $a_0(t,p)$.
Similarly, we denote by $a_1(t,p)$ the unique root w.r.t.\ $\xi$ of $v_1(t,p,\xi,1)=\mu-r$.
Setting
\begin{equation*}
z(t,p) := \begin{cases} (a_0(t,p),0), & \theta\le A(t,p)/\kappa-1,\\
(a_1(t,p),1), &\theta\ge B(t,p)/\kappa-1, \\
r(t,p), &\text{otherwise},
\end{cases}
\end{equation*}
we obtain the following representation of the candidate for an optimal investment-reinsurance strategy $(\xi^\star,b^\star)=(\xi^\star_t,b^\star_t)_{t\in[0,T]}$:
\begin{equation}\label{optstr}
(\xi^\star_t,b^\star_t) := z(t,p_{t-}),\quad t\in[0,T].
\end{equation}
Notice that the strategy $(\xi^\star,b^\star)$ can only jump at the claim arrival times due to the dependency on the filter process $(p_t)_{t\ge0}$. 

\subsection{Verification}
This section is devoted to a verification theorem to ensure that the solution of the stated generalized HJB equation yields the value function (see Theorem~\ref{veri}). We also demonstrate an existence theorem of the solution of the HJB equation (see Theorem~\ref{existenceHJB}).  Both proofs can be found in the appendix.

\begin{theorem}\label{veri}
	Suppose there exists a bounded function $h:[0,T]\times\Delta_m\to(0,\infty)$ such that $t\mapsto h(t,p)$ is Lipschitz on $[0,T]$ for all $p\in\Delta_m$, $p\mapsto h(t,p)$ is continuous on $\Delta_m$ for all $t\in[0,T]$ and $h$ satisfies the generalized HJB equation \eqref{HJBg}	for all $(t,p)\in[0,T)\times\Delta_m$ with boundary condition
	\begin{equation}\label{hHJBbcond}
	h(T,p) = 1,\quad p\in\Delta_m.
	\end{equation}
	Then
	\begin{equation*}
	V(t,x,p)= -e^{-\alpha x e^{r(T-t)}}h(t,p),\quad (t,x,p)\in[0,T]\times\R\times\Delta_m,
	\end{equation*}
	and  $(\xi^\star,b^\star)=(\xi^\star_s,b^\star_s)_{s\in[t,T]}$ with $(\xi^\star_s,b^\star_s)$ given by~\eqref{optstr} (with $g$ replaced by $h$ in $A(s,p)$ and $B(s,p)$) is an optimal feedback strategy for the given optimization problem~\eqref{P}, i.e.\ $V(t,x,p) = V^{\xi^\star,b^\star}(t,x,p)$.
\end{theorem}

\subsection{Existence result for the value function}
\label{existencevalue}

We now show that there exists a function $h:[0,T]\times\Delta_m\to(0,\infty)$ satisfying the conditions stated in Theorem~\ref{veri}. For this purpose let 
\begin{equation}\label{g}
g(t,p) := \inf_{(\xi,b)\in{\mathcal U}[t,T]}g^{\xi,b}(t,p),
\end{equation}
with
\begin{equation}\label{gxib}
\begin{aligned}
g^{\xi,b}(t,p) := \Eop^{t,p}\bigg[\exp\bigg\{&-\int_t^T \alpha e^{r(T-s)}\big((\mu-r)\,\xi_s+c(b_s)\big)ds -\int_t^T\alpha\sigma e^{r(T-s)}\xi_sdW_s\\
&+\int_t^T \int_E \alpha \big( b_s y + \xi_s z \mathds{1}_{(L,\infty)}(y)\big)e^{r(T-s)}\Psi(ds, d(y,z))\bigg\}\bigg],
\end{aligned}
\end{equation}
where $\Eop^{t,p}$ denotes the conditional expectation given $(p_t,q_t)=(p,q)$. The next lemma summarizes useful properties of $g$. A proof can be found in the appendix.

\begin{lemma}\label{propg}
	The function $g$ defined by~\eqref{g} has the following properties:
	\begin{enumerate}
		\item $g$ is bounded on $[0,T]\times\Delta_m$ by a constant $0<K_1<\infty$ and $g>0$.
		\item $g^{\xi,b}(t,p) = \sum_{j=1}^m p_j g^{\xi,b}(t,e_j)$ for all $(t,p)\in[0,T]\times\Delta_m$ and $(\xi,b)\in{\mathcal U}[t,T]$.
      \item $g^{\xi,b}(t,J(p,y)) = \sum_{j=1}^m \frac{f_j(y)p_j}{\sum_{k=1}^m f_k(y)p_k} g^{\xi,b}(t,e_j)$ for all $(t,p)\in[0,T]\times\Delta_m$ and $(\xi,b)\in{\mathcal U}[t,T]$.
		\item $\Delta_m\ni p\mapsto g(t,p)$ is concave for all $t\in[0,T]$.
		\item $[0,T]\ni t\mapsto g(t,p)$ is Lipschitz on $[0,T]$ for all $p\in\Delta_m$.
	\end{enumerate}
\end{lemma}

Notice that $e_j$ denotes the $j$th unit vector. We are now in the position to show the following existence result of a solution of the generalized HJB equation.

\begin{theorem}\label{existenceHJB}
	The value function of  problem~\eqref{P} is given by
	\begin{equation*}
	V(t,x,p) = -e^{-\alpha x e^{r(T-t)}}g(t,p),\quad (t,x,p)\in[0,T]\times\R\times\Delta_m,
	\end{equation*}
	where $g$ is defined by~\eqref{g} and satisfies the generalized HJB equation \eqref{HJBg} for all $(t,p)\in[0,T)\times\Delta_m$ with boundary condition $g(T,p)=1$ for all $p\in\Delta_m$. 
	Furthermore, $(\xi^\star,b^\star)=(\xi^\star_s,b^\star_s)_{s\in[t,T]}$ with $(\xi^\star_s,b^\star_s)$ given by~\eqref{optstr} is the optimal investment and reinsurance strategy of the optimization problem~\eqref{P}. 
\end{theorem}

\section{Comparison results}\label{sec:comp}

\subsection{Case of independent financial and insurance risks}

In this section we present a comparison result of the optimal strategy given in Theorem~\ref{existenceHJB} and the one in the case of independent financial and insurance risks. In this case the price process of the risky asset has no jumps if an insurance claim exceed the threshold $L$, i.e.\ the price process of the risky asset evolves according to a geometric Brownian motion. 
Throughout this section, we suppose that $K$ is large.
We write $(\tilde\xi^\star,\tilde b^\star)$ for the optimal investment and reinsurance strategy in the case of no interdependencies between the financial and insurance market as describe above. 
We obtain the special solution (cp. \cite[Ch.\,6]{gl20})
\begin{align*}
\tilde\xi^\star_t &= \frac{\mu-r}{\sigma^2}\frac{1}{\alpha}e^{-r(T-t)},\\
\tilde b^\star_t &= \tilde b(t,p_{t-}),
\end{align*}
where
\begin{equation*}
\tilde b(t,p) := \begin{cases} 0, & \theta\le \tilde A(t,p)/\kappa-1,\\
1, &\theta\ge \tilde B(t,p)/\kappa-1, \\
\tilde r(t,p), &\text{otherwise},
\end{cases}
\end{equation*}
with
\begin{align*}
  \tilde\gamma(t,p,b) &:= \lambda\sum_{k=1}^m p_k \int_0^\infty y \frac{g(t,J(p,y))}{g(t,p)} e^{\alpha b y e^{r(T-t)}} f_k(y)dy, \\
   \tilde A(t,p) &:= \tilde\gamma(t,p,0), \\
   \tilde B(t,p) &:= \tilde\gamma(t,p,1),
\end{align*}
and $\tilde r(t,p)$ is the unique root of $\tilde\gamma(t,p,b)=(1+\theta)\kappa$ w.r.t.\ $b$.
The next theorem provides a comparison of the optimal investment strategies $\xi^\star$ and $\tilde\xi^\star$.

\begin{theorem}\label{comparison}
   For any $t\in[0,T]$ it holds $  \xi^\star_t \le \tilde\xi^\star_t.$
\end{theorem}

\begin{proof}
  Fix $t\in[0,T]$. Note that the first order condition of $\tilde\xi^\star$ is
  \begin{equation*}
  \alpha\sigma^2 e^{r(T-t)}\xi=\mu-r,
  \end{equation*}
  where the left-hand side is always less than $v_1(t,p,\xi,b)$ from~\eqref{foc} and crosses $\mu-r$ from below.
  Consequently, $\xi_t^\star\le \frac{\mu-r}{\sigma^2}\frac{1}{\alpha}e^{-r(T-t)} $. 
\end{proof}

The theorem says that it is always optimal to invest more money into the risky asset in the absence of interdependencies between financial and insurance risks than in the presence of dependencies. This is not surprising since the interdependency in our model may only imply some downward jumps of the risky asset. A negative investment into the financial market can  be used to hedge against claims.

\subsection{Case of complete information}

First note that the case with complete information is always a special case of our general model. We obtain this case when the prior is concentrated on a single value.
In order to state the optimal strategy in the complete information case, we define for any $t\in[0,T]$ and $(\xi,b)\in\R^2$
\begin{align*}
v_1^F(t,\xi,b) &:= \alpha\sigma^2e^{r(T-t)}\xi + \lambda \int_L^\infty e^{\alpha b y e^{r(T-t)}}F(dy)M_Z^\prime\big(\alpha e^{r(T-t)}\xi\big),\\
v_2^F(t,\xi,b) &:= \lambda \int_0^\infty y e^{\alpha b y e^{r(T-t)}}\int_{(0,1)}e^{\alpha\xi z \mathds{1}_{(L,\infty)}(y)e^{r(T-t)}}Q(dz)F(dy),
\end{align*}
for some distribution $F$ on $(0,\infty)$.
Furthermore, we denote by $r^F(t)=(r^F_1(t),r^F_2(t))$ the unique root w.r.t.\ $(\xi,b)$ of
\begin{equation}\label{full:foc}
\begin{aligned}
v_1^F(t,\xi,b) &= \mu-r\\
v_2^F(t,\xi,b) &= (1+\theta)\kappa,
\end{aligned}
\end{equation}
which exists, and we define
\begin{equation*}
A_F(t) := v_2^F(t,r^F_1(t),0), \quad B_F(t) := v_2^F(t,r^F_1(t),1).
\end{equation*}
Moreover, $a_0^F(t)$ denotes the unique root w.r.t.\ $\xi$ of $v_1^F(t,\xi,0)=\mu-r$ and $a_1^F(t)$ the unique root w.r.t.\ $\xi$ of $v_1^F(t,\xi,1)=\mu-r$.
By the same line of arguments as in Proposition~\ref{candidates}, we obtain under the notation above that the optimal reinsurance strategy $(\xi_F^\star,b^\star_{F})=(\xi_F^\star(t),b^\star_{F}(t))_{t\in[0,T]}$ in the case of complete information is given by
\begin{equation}\label{full:xibstart}
(\xi_F^\star(t),b^\star_F(t)) := 
\begin{cases} (a_0^F(t),0), & \theta\le A_F(t)/\kappa -1, \\
(a_1^F(t),1), &\theta \ge B_F(t)/\kappa-1, \\ 
r^F(t), &\text{otherwise}.
\end{cases}
\end{equation}
Note that $r_1^F(t)$, $r_2^F(t)$, $a_0^F(t)$, $a_1^F(t)$, $A_F(t)$ and $B_F(t)$ are continuous in $t$. Consequently, the optimal strategies $\xi^\star_F$ and  $b^\star_F$ is continuous. Moreover, $(\xi^\star_F,b^\star_F)$ is deterministic and can be calculated easily.

We will now compare the strategies. In order to do so, we assume throughout this section that
\begin{equation*}
F_1(x)\ge F_2(x)\ge \ldots \ge F_m(x)
\end{equation*}
for all $x\in\R$. That is, the claim sizes are ordered stochastically as follows:
\begin{equation*}
Y|\vartheta=1 \preceq_{\textup{st}} Y|\vartheta=2 \preceq_{\textup{st}}\ldots\preceq_{\textup{st}} Y|\vartheta=m,
\end{equation*}
where $\preceq_{\textup{st}}$ denotes the usual stochastic order. This assertion is equivalent to
\begin{equation*}
\int_0^\infty g(y) f_1(y)dy\le\int_0^\infty g(y) f_2(y)dy\le \ldots \le \int_0^\infty g(y) f_m(y)dy
\end{equation*}
for all increasing functions $g$, for which the expectations exist, compare Theorem 1.2.8 in \cite{MuellerStoyan2002}.

First of all we derive bounds for the optimal strategy which can be calculated apriori, i.e.\ independent of the filter process $(p_t)_{t\ge0}$.
For this determination, we introduce the following terms. For any $t\in[0,T]$ and $(\xi,b)\in\R^2$, we set
\begin{align*}
v_1^{\min}(t,\xi,b) &:= \alpha\sigma^2e^{r(T-t)}\xi + \lambda \int_L^\infty e^{\alpha b y e^{r(T-t)}}f_1(y)dy M_Z^\prime\big(\alpha \xi e^{r(T-t)}\big),\\
v_1^{\max}(t,\xi,b) &:= \alpha\sigma^2e^{r(T-t)}\xi + \lambda \int_L^\infty e^{\alpha b y e^{r(T-t)}}f_m(y)dy M_Z^\prime\big(\alpha \xi e^{r(T-t)}\big),
\end{align*}
For some fixed $t\in[0,T]$, we denote by $r^{\min}_1(t)$ the unique root of $v_1^{\min}(t,\xi,b) = \mu-r$ and by $r^{\max}_1(t)$ the unique root of $v_1^{\max}(t,\xi,b) = \mu-r$,  which exist by the same line of arguments as in Proposition~\ref{candidates}.
The announced a-priori-bounds are a direct consequence of the following result.

\begin{proposition}\label{pr:aprioribounds}
   For any $(t,p)\in[0,T]\times\Delta_m$, we have for $v_1$ from \eqref{foc}
   \begin{align*}
   v_1^{\min}(t,\xi,b) &\le v_1(t,p,\xi,b) \le v_1^{\max}(t,\xi,b)\quad\text{for all }(\xi,b)\in\R\times\R_+.
   \end{align*}
\end{proposition}

\begin{proof}
   Choose some $(t,p)\in[0,T]\times\Delta_m$ and $(\bar\xi,\bar b)\in\R\times\R_+$. For any $(\xi,b)\in\mathcal U[t,T]$, an application of Lemma~\ref{propg}~(b) and~(c) yields
   \begin{align*}
   &\sum_{k=1}^{m}p_k\int_L^\infty g^{\xi,b}(t,J(p,y))e^{\alpha\bar b y e^{r(T-t)}}f_k(y)dy \\
   &= \sum_{j=1}^{m}p_j g^{\xi,b}(t,e_j) \int_L^\infty\frac{\sum_{k=1}^{m}p_kf_k(y)}{\sum_{\ell=1}^{m}p_\ell f_\ell(y)} e^{\alpha \bar b y e^{r(T-t)}}f_j(y)dy 
   \le g^{\xi,b}(t,p) \int_L^\infty e^{\alpha \bar b y e^{r(T-t)}}f_m(y)dy,
   \end{align*}
   which yields $v_1(t,p,\bar\xi,\bar b) \le v_1^{\max}(t,\bar\xi,\bar b)$ by dividing by $ g^{\xi,b}$, multiplying both sides by \linebreak $\lambda M_Z^\prime(\alpha \bar\xi e^{r(T-t)})$ and by adding $\alpha\sigma^2e^{r(T-t)}\bar\xi$.
   The inequality $v_1^{\min}(t,\bar \xi,\bar b) \le v_1(t,p,\bar\xi,\bar b)$ is obtained in the same way.
\end{proof}

The proposition directly implies the following corollary:

\begin{corollary}\label{co:aprioribounds}
   The optimal investment strategy $\xi^\star$ from Theorem~\ref{existenceHJB} has the following bounds for $ t\in[0,T]$:
   \begin{eqnarray*}
   r_1^{\max}(t)  \le \xi^\star_t  \quad\text{if } b^\star_{F_1}(t)=b_t^\star,\\
  \xi^\star_t  \le   r_1^{\min}(t)  \ \quad\text{if } b^\star_{F_m}(t)=b_t^\star.
   \end{eqnarray*}
  
\end{corollary}

 The next theorem is now the main statement of this section. It provides a comparison of the optimal investment strategy to the optimal one in the case of complete information, where the unknown claim size distribution is replaced by their expectation. It turns out that in the latter case the amount which is invested is higher if the retention is the same. In this sense the complete information case provides upper bounds.
 
\begin{theorem}\label{th:comparison}
   Let $(\xi_F^\star,b_F^\star)$ be the function given in~\eqref{full:xibstart} and suppose the insurance company does invest into the financial market, i.e.\ $\xi_t^\star >0$ for all $t\in [0,T]$. Then  if $b_t^\star = b^\star_{\bar F_{p_{t-}}}$ we obtain for $t\in[0,T]$
   \begin{equation*}
   \xi_t^\star \le \xi^\star_{\bar F_{p_{t-}}}(t),\quad \mbox{with }  \bar F_p(dy) := \sum_{k=1}^m p_kf_k(y)dy.
   \end{equation*}
\end{theorem}

\begin{proof}
   Let us fix $(t,p)\in[0,T]$ and $(\bar\xi,\bar b)\in\R\times\R_+$. From the proof of Proposition~\ref{pr:aprioribounds}, we know already that
   \begin{equation*}
   \sum_{k=1}^{m}p_k\int_L^\infty g^{\xi,b}(t,J(p,y))e^{\alpha\bar b y e^{r(T-t)}}f_k(y)dy
   = \sum_{j=1}^{m}p_j g^{\xi,b}(t,e_j) \int_L^\infty e^{\alpha \bar b y e^{r(T-t)}}f_j(y)dy
   \end{equation*}
   for all $(\xi,b)\in\widetilde{\mathcal{U}}[t,T]$, where $\widetilde{\mathcal{U}}[t,T]$ denotes the set of all admissible strategies $\mathcal{U}[t,T]$ restricted to positive investment strategies. The integrand of
   \begin{align*}
   g^{\xi,b}(t,p) = \Eop^{t,p}\bigg[\exp\bigg\{&-\int_t^T\alpha e^{r(T-s)}\big((\mu-r)\xi_s+c(b_s)\big) ds -\int_t^T\alpha e^{r(T-s)}\xi_s dW_s\\
   &+\sum_{n=1}^{N_{T-t}}\alpha \big(b_{T_n}Y_n+\xi_{T_n} Z_n\mathds{1}_{(L,\infty)}(Y_n)\big)e^{r(T-T_n)}\bigg\}\bigg]
   \end{align*}
   is increasing in $Y_n$ (due to the positivity of $\xi_t$ for all $t\in[0,T]$) and hence $g^{\xi,b}(t,e_1)\le \ldots\le g^{\xi,b}(t,e_m)$.
   Therefore, by Lemma~\ref{propg}~(b) as well as Lemma~\ref{le:ineqsum}, we get
   \begin{align*}
   \sum_{j=1}^{m}p_j g^{\xi,b}(t,e_j) \int_L^\infty e^{\alpha \bar b y e^{r(T-t)}}f_j(y)dy
   \ge g^{\xi,b}(t,p) \int_L^\infty e^{\alpha \bar b y e^{r(T-t)}}\sum_{j=1}^m p_j f_j(y)dy.
   \end{align*}
   In summary, we have 
   \begin{equation*}
   \sum_{k=1}^{m}p_k\int_L^\infty g^{\xi,b}(t,J(p,y))e^{\alpha\bar b y e^{r(T-t)}}f_k(y)dy
   \ge g^{\xi,b}(t,p) \int_L^\infty e^{\alpha \bar b y e^{r(T-t)}}\bar F_p(dy),
   \end{equation*}
   for all $(\xi,b)\in\widetilde{\mathcal{U}}[t,T]$, which yields $v_1(t,p,\bar\xi,\bar b)\ge v_1^{\bar F_p}(t,\bar\xi,\bar b)$ by the same argumentation as in the proof of Proposition~\ref{pr:aprioribounds}. Therefore, we get $\xi_t^\star \le \xi^\star_{\bar F_{p_{t-}}}(t)$ under the assumptions $\xi_t^\star>0$.
\end{proof}

\subsection{Numerical results}
We have seen in the last subsection that it is easy to compute the optimal strategy in the case of full information and that this yields in some cases a bound on the optimal strategy in the case of  incomplete information.  In particular when we set $r=0$ then the strategy obtained through \eqref{full:xibstart} is a constant and does not depend on time, only on the final time horizon.

We have computed the optimal strategy in the case of full information for the following data: The volatility of the financial market is $\sigma=0.4$, the drift $\mu=0.3$ and the interest rate $r=0$. The claim arrival intensity is $\lambda=10$ and the claim sizes are exponentially distributed with parameter $\varrho=0.1$, i.e. $Y\sim Exp(0.1)$. Note that the moment generating function of the exponential distribution does exist only for $\alpha \in (0,\varrho)$. Thus, for all integrals to exist we have to make sure that $\alpha < \varrho$. Hence, we choose $\alpha=0.05$ which means that we are close to the risk-sensitive case. For $Z$ we choose a uniform distribution on $(0,1)$. The expected amount of claims per year in this model is $\Eop N_1 \Eop Y= 100$, so we should choose $(1+\theta)\kappa >100$. Indeed since the premium income itself  is below $(1+\theta)\kappa $ we set $(1+\theta)\kappa =350$. We compute now the optimal investment and reinsurance strategy for different level $L$. Note that the expected claim size is $10$. The larger $L$, the smaller will be the constructed dependency between the markets. For $L\to \infty$ we obtain independence. Figure \ref{fig:influenceL} and shows the results.


\begin{figure}
   \begin{center}
      \includegraphics[width=0.95\textwidth]{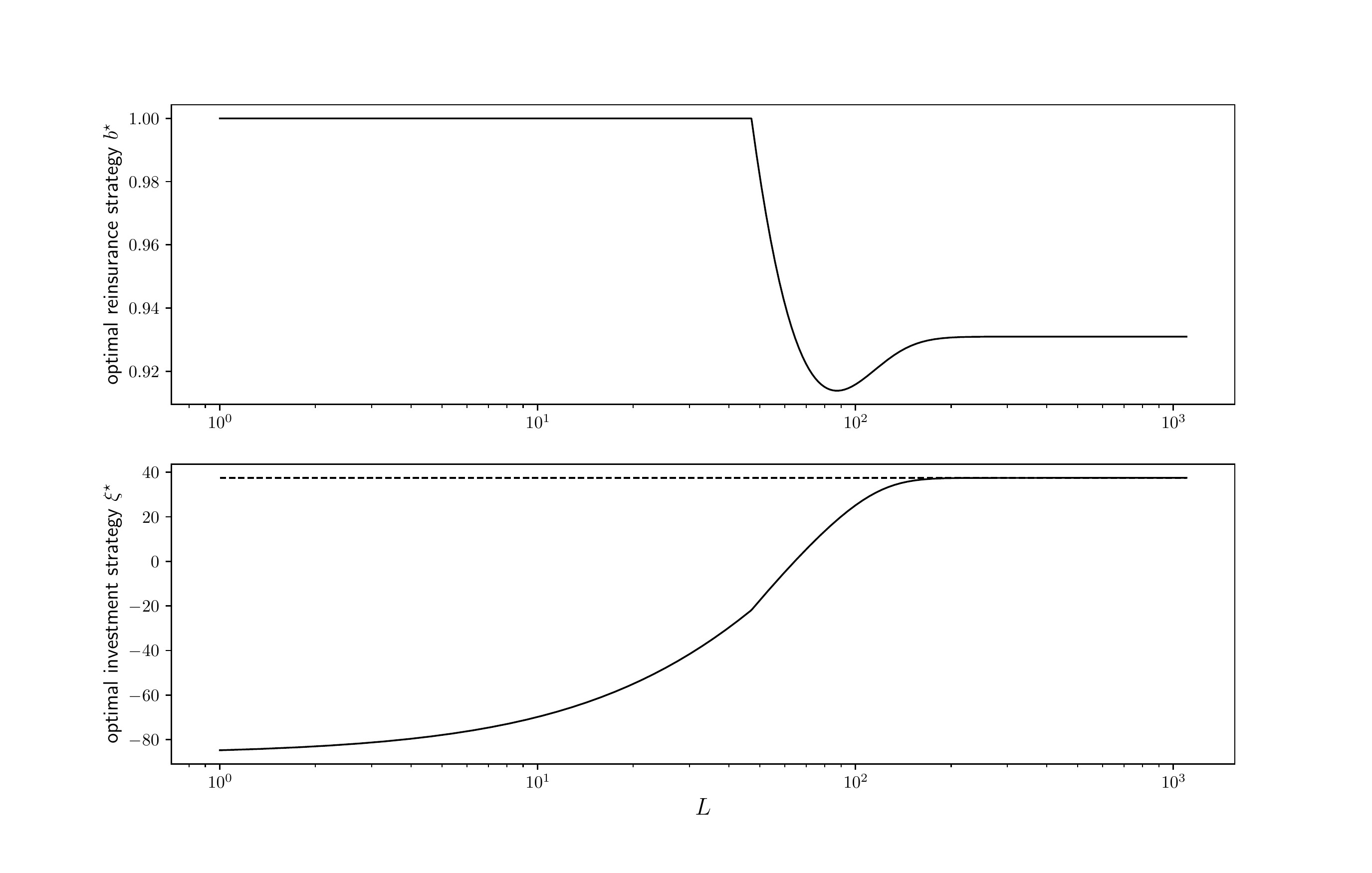}
   \end{center}
   \caption{Optimal strategy in the case of complete observation as a function of $L$ with logaritmically scaled $x$-axis.}
   \label{fig:influenceL}
\end{figure}

So what we obviously see here is that with increasing $L$ the investment is increasing. This may be expected since there will be less drop downs in the financial market when $L$ is large.  However, what is surprising is the following observation:  In the independent case the optimal investment with these parameters is $\xi^\star=\frac{\mu}{\alpha\sigma^2}=37.5$ and for $L\to\infty$ we can see a convergence. But even if $L=100$ which means that the threshold which produces the correlation is 10 times as high as an expected claim, i.e.\ very unlikely to occur (the probability indeed is $4.5^{-5}$) the investment in the risky asset is only $25.19$ compared to $37.5$. Thus, the insurance company is very conservative.  Of course we have a risk-sensitive criterion here, but nevertheless the impact of the dependency is amazing. For $L$ below $64.35$ there is a negative investment into the financial market.
The insurance company then uses the dependence to hedge against claims by shortselling stocks.
In the case of $L\to0$, the optimal investment converges to $-86.57$.
For smaller $L$ there is indeed no reinsurance.  
For $L\to\infty$ the value stabilizes around $b^\star=0.93$, i.e.\  only  $7\,\%$ of the claims are covered by reinsurance.

 In total, the conclusion that we draw here is that in this simple model introducing only a small correlation between claim sizes and behavior of the financial market has already a severe impact on the optimal investment strategy.


\section{Appendix}\label{sec:app}

\subsection{Clarke's generalized subdifferential}
The following definition and results are taken from Section~2.1 in \cite{Clarke1983}, where we restrict ourself to some univariate function by $f:\R\to\R$, which is sufficient for this paper.

\begin{definition}[\cite{Clarke1983}, p.\,25]\label{def:gendirder}
   Let $x\in\R$ be a given point and let $v\in\R$. Moreover, let $f$ be Lipschitz near $x$. Then the \emph{generalized directional derivative} of $f$ at $x$ in the direction $v$, denoted by $f^\circ(x;v)$, is defined by
   \begin{equation*}
   f^\circ(x;v) = \limsup_{y\to x, h\downarrow 0}\frac{f(y+h\,v)-f(y)}{h}.
   \end{equation*}
\end{definition}

\begin{definition}[\cite{Clarke1983}, p.\,27]\label{def:Clarksub}
   Let $f$ be Lipschitz near $x$. Then \emph{Clarke's generalized subdifferential} of $f$ at $x$, denoted by $\partial^C f(x)$, is given by
   \begin{equation*}
   \partial^Cf(x):=\big\{\xi\in\R: f^\circ(x;v)\ge \xi v \text{ for all }v\in\R\big\}.
   \end{equation*}
\end{definition}

\begin{proposition}[\cite{Clarke1983}, Prop.\,2.2.4]\label{pr:propClarkediff}
   If $f$ is strictly differentiable at $x$, then $f$ is Lipschitz near $x$ and $\partial^C f(x) = \{f^\prime(x)\}$. Conversely, if $f$ is Lipschitz near $x$ and $\partial^C f(x)$ reduces to a singleton $\{\zeta\}$, then $f$ is strictly differentiable at $x$ and $f^\prime(x)=\zeta$.
\end{proposition}

\begin{theorem}[\cite{Clarke1983}, Thm.\,2.5.1]\label{th:genCgco}
   Let $f$ be Lipschitz near $x$ and let $S$ be an arbitrary set of Lebesgue-measure $0$ in $\R$. Moreover, the set of points, at which the function $f$ is not differentiable, is denoted by $\Omega_f$.
   Then
   \begin{equation*}
   \partial^C f(x) = co\Big\{\lim_{n\to\infty} f^\prime(x_n): x_n\to x, x_n\notin S, x_n\notin\Omega_f\Big\}.
   \end{equation*}
\end{theorem}

\subsection{Auxiliary Results}

From now on, we denote by $f:[0,T]\times\R\to\R$ the function which is defined by
\begin{equation}\label{f}
f(t,x) := -e^{-\alpha x e^{r(T-t)}}.
\end{equation}

\begin{lemma}\label{Qxibt}
   Let $t\in[0,T]$ and let $(\xi,b)\in\mathcal{U}[0,T]$ be an arbitrary admissible strategy. We set
   \begin{equation}\label{density}
   \begin{aligned}
L^{\xi,b}_t &:= \exp\bigg\{ -\int_0^t \alpha\sigma e^{r(T-s)}\xi_s dW_s -\frac12 \int_0^t \alpha^2\sigma^2e^{2r(T-s)}\xi_s^2 ds \\
&\;\quad+\int_0^t\int_E \alpha(b_s y + \xi_s z \mathds{1}_{(L,\infty)}(y)) e^{r(T-s)}\Psi(ds, d(y,z)) + \lambda t \\
&\;\quad -\int_0^t\lambda \sum_{k=1}^m p_k(s) \int_0^\infty e^{\alpha b_s y e^{r(T-s)}}\int_{(0,1)} e^{\alpha \xi_s z\mathds{1}_{(L,\infty)}(y) e^{r(T-s)}}Q(dz)f_k(y)dy ds\bigg\}.
\end{aligned}
\end{equation}
   Then, a possibly substochastic  measure on $(\Omega,\mathcal{G}_t)$ is defined by $\Q^{\xi,b}_t(A):=\int_A L^{\xi,b}_t d\Pop$, $A\in\mathcal{G}_t$, for every $t\in[0,T]$, i.e.\ $\frac{d\Q^{\xi,b}_t}{d\Pop}:=L^{\xi,b}_t$. The  measures $\Q^{\xi,b}_t$ and $\Pop$ are equivalent.
\end{lemma}

\begin{proof}
   First, we show that $(L_t^{\xi,b})_{t\ge0}$ is the Dol\'{e}ans-Dade exponential of the martingale $(Z_t)_{t\ge0}$ defined by
   \begin{equation*}
   Z_t := -\int_0^t \alpha \sigma e^{r(T-s)}\xi_sdW_s + \int_0^t \int_E \Big(e^{\alpha(b_s y+ \xi_s z \mathds{1}_{(L,\infty)}(y))e^{r(T-s)}}-1\Big)\hat{\Psi}(ds,d(y,z)).
   \end{equation*}
   That is,
   \begin{equation*}
   L_t^{\xi,b} = \mathcal{E}(Z_t) = e^{Z_t - \frac12\int_0^t \alpha^2\sigma^2e^{2r(T-s)}\xi_s^2ds}\prod_{0<s\le t}(1+\Delta Z_s)e^{-\Delta Z_s},
   \end{equation*}
   where
   \begin{align*}
   \prod_{0<s\le t}\!(1\!+\!\Delta Z_s)e^{-\Delta Z_s} &= \exp\bigg\{\int_0^t\!\int_E \alpha(b_s y+ \xi_s z \mathds{1}_{(L,\infty)}(y))e^{r(T-s)}\Psi(ds,d(y,z))\bigg\} \\
   &\quad\times\! \exp\bigg\{\!\!-\!\int_0^t\!\int_E\! \Big(\!\exp\Big\{\alpha(b_s y+ \xi_s z \mathds{1}_{(L,\infty)}(y))e^{r(T-s)}\Big\}\!-\!1\Big)\Psi(ds,d(y,z))\bigg\}.
   \end{align*}
   This implies the announced representation \eqref{density} of $(L_t^{\xi,b})_{t\ge0}$ since $\hat\Psi-\Psi = \hat\nu$. As $(L_t^{\xi,b})_{t\ge0}$  is a non-negative local martingale, it is a supermartingale and hence $\Eop L_t^{\xi,b} \le 1$ for all $t\ge 0$. 
\end{proof}

\begin{lemma}\label{fbounded}
   Let $\xib\in \mathcal{U}[0,T]$ and let $\prT{L^{\xi,b}}$ be the density process given by~\eqref{density}.
   Then there exists a constant $0<K_2<\infty$ such that 
   \begin{equation*}
   \frac{\big|f(t,X^{\xi,b}_t)\big|}{L^{\xi,b}_t}\le K_2\quad\Pop\text{-a.s.}
   \end{equation*}
   for all $t\in[0,T]$.
\end{lemma}

\begin{proof}
   Fix $t\in[0,T]$ and $\xib\in  \mathcal{U}[0,t]$. 
   Using Theorem V.52 in \cite{Protter2005}, the unique solution of \eqref{wealth} is 
   \begin{align*}
   X_t^{\xi,b} &= x_0e^{rt} + \int_0^t e^{r(t-s)}\big((\mu-r)\xi_s+c(b_s)\big)ds + \int_0^t \sigma e^{r(t-s)}\xi_sdW_s \\
   &\quad + \int_0^t \int_E e^{r(t-s)} \big(b_s y + \xi_s z\mathds{1}_{(L,\infty)}(y)\big)\Psi(ds,d(y,z))
   \end{align*}
   Hence
   \begin{align*}
   &\frac{\big|f(t,X^{\xi,b}_t)\big|}{L^{\xi,b}_t} = \exp\bigg\{-\alpha x_0 e^{rT} -\int_0^t \alpha e^{r(T-s)}\Big( (\mu-r)\xi_s +c(b_s) - \frac12\alpha\sigma^2 e^{r(T-s)}\xi_s^2\Big) ds  \\
   &\;+ \int_0^t\lambda \sum_{k=1}^m p_k(s) \int_0^\infty e^{\alpha b_s y e^{r(T-s)}}\int_{(0,1)} e^{\alpha \xi_s z\mathds{1}_{(L,\infty)}(y) e^{r(T-s)}}Q(dz)f_k(y)dy ds -\lambda t\bigg\} \\
   &\le\exp\bigg\{\bigg(\alpha e^{|r|T}\big(|\mu-r|K+(2+\eta+\theta)\kappa\big)+\frac12\alpha^2\,\sigma^2\, e^{2|r|T}K^2\\
   &\qquad\qquad + \lambda\sum_{k=1}^m M_k\big(\alpha e^{|r|T}\big)M_Z\big(\alpha K e^{|r|T}\big)\bigg)T\bigg\}=: K_2,
   \end{align*}
   where $0<K_2<\infty$ is independent of $t\in[0,T]$ as well as $\xib$.
\end{proof}

For convenience we define
\begin{equation}\label{H}
\mathcal{H} h(t,p;\xi,b) := \mathcal{L} h(t,p;\xi,b) + h_t(t,p)
\end{equation}
for all functions $h:[0,T]\times\Delta_m\to(0,\infty)$ and $(\xi,b)\in \R\times [0,1]$, where the right-hand side is well-defined.
Using this notation, the generalized HJB equation~\eqref{HJBg} can be written as
\begin{equation}\label{HHJB}
0= \inf_{(\xi,b)\in[-K,K]\times[0,1]}\{\mathcal{H} g(t,p;\xi,b)\}
\end{equation}
at those points $(t,p)$ with existing $g_t(t,p)$.

\begin{lemma}\label{characG}
	Suppose that $(\xi,b)\in \mathcal U[0,T]$ is an arbitrary strategy and $h:[0,T]\times\Delta_m\to(0,\infty)$ is a bounded function such that $t\mapsto h(t,p)$ is absolutely continuous on $[0,T]$ for all $p\in\Delta_m$ and $p\mapsto h(t,p)$ is continuous on $\Delta_m$ for all $t\in[0,T]$. Then, the function $G:[0,T]\times\R\times\Delta_m\to\R$ defined by
	\begin{equation*}
	G(t,x,p) := -e^{-\alpha x e^{r(T-t)}}h(t,p)
	\end{equation*}
	satisfies
	\begin{equation*}
	d G(t,X^{\xi,b}_t,p_t)  = -e^{-\alpha X^{\xi,b}_t e^{r(T-t)}}\mathcal{H} h(t,p_t;\xi_t,b_t)dt + d\eta^{\xi,b}_t, \quad t\in[0,T],
	\end{equation*}
	where $(\eta^{\xi,b}_t)_{t\in[0,T]}$ is a martingale w.r.t.\ ${\mathfrak G}$ and we set $\mathcal{H} h(t,p;\xi,b)$ zero at those points $(t,p)$ where $h_t$ does not exist.
\end{lemma}

\begin{proof}
	 Let $\xib\in \mathcal U[0,T]$ and $h:[0,T]\times\Delta_m\to(0,\infty)$ be some function satisfying the conditions stated in the lemma and bounded with constant $0<K_0<\infty$. Applying the product rule  to $G\big(t,X^{\xi,b}_t,p_t\big)=f\big(t,X^{\xi,b}_t\big)h(t,p_t)$, we get
	\begin{equation*}
	d G\big(t,X^{\xi,b}_t,p_t\big) 
	= h(t,p_{t-})df\big(t,X^{\xi,b}_t\big) + f\big(t,X^{\xi,b}_{t-}\big) d h(t,p_t)  + d\big[f\big(\cdot,X^{\xi,b}_\cdot\big),h(\cdot,p_\cdot)\big]_t 
	\end{equation*}
	and hence
	\begin{equation}\label{G}
	\begin{aligned}
	&d G\big(t,X^{\xi,b}_t,p_t\big) \\
	&= f\big(t,X^{\xi,b}_t\big)h(t,p_t)\bigg(\alpha e^{r(T-t)}\Big(\frac12\alpha\sigma^2 e^{r(T-t)}\xi_t^2 - (\mu-r)\xi_t - c(b_t)\Big) \\
	&\quad + \lambda\sum_{k=1}^m p_k(t)\int_0^\infty e^{\alpha b_t y e^{r(T-t)}}\int_{(0,1)}\!e^{\alpha \xi_t z \mathds{1}_{(L,\infty)}(y)e^{r(T-t)}}Q(dz)f_k(y)dy - \lambda\bigg)dt \\
	&\quad-f\big(t,X^{\xi,b}_{t-}\big) h(t,p_{t-}) \alpha\sigma e^{r(T-t)} \xi_tdW_t \\
   &\quad+\int_0^\infty f\big(s,X^{\xi,b}_{t-}\big)h(t,p_{t-})\big(e^{\alpha b_t y e^{r(T-t)}}e^{\alpha \xi_t z \mathds{1}_{(L,\infty)}(y) e^{r(T-t)}}-1\big)\hat\Psi(dt,d(y,z)) \\
	&\quad+f\big(t,X^{\xi,b}_{t}\big) \bigg(h_t(t,p_t) - \lambda h(t,p_t) + \lambda\sum_{k=1}^m p_k(t)\int_0^\infty h(t,J(p_t,y))f_k(y)dy\bigg)dt \\
	&\quad+\int_0^\infty f\big(t,X^{\xi,b}_{t-}\big)\big(h(t,J(p_{t-},y))-h(t,p_{t-})\big)\hat{\Psi}(dt,dy,(0,1)) \\
	&\quad + d\big[f\big(\cdot,X^{\xi,b}_\cdot\big),h(\cdot,p_\cdot)\big]_t.
	\end{aligned}
	\end{equation}
Using the introduced compensated random measure $\hat{\Psi}$ the variation becomes
\begin{align*}
&d\big[f\big(\cdot,X^{\xi,b}_\cdot\big),h(\cdot,p_\cdot)\big]_t \\
&= \int_E f\big(t,X^{\xi,b}_{t-}\big)\big(h(t,J(p_{t-},y))-h(t,p_{t-})\big)\Big(e^{\alpha b_t y e^{r(T-t)}}e^{\alpha \xi_t z \mathds{1}_{(L,\infty)}(y) e^{r(T-t)}} -1\Big)\hat{\Psi}(dt, d(y,z)) \\
&\quad + \lambda f\big(t,X_t^{\xi,b}\big)\sum_{k=1}^m p_k(t) \int_0^\infty h(t,J(p_t,y))e^{\alpha b_t y e^{r(T-t)}}\int_{(0,1)}\!e^{\alpha \xi_t z \mathds{1}_{(L,\infty)}(y)e^{r(T-t)}}Q(dz)f_k(y)dydt \\
&\quad - \lambda f\big(t,X_t^{\xi,b}\big)h(t,p_t)\sum_{k=1}^m p_k(t) \int_0^\infty e^{\alpha b_t y e^{r(T-t)}}\int_{(0,1)}\!e^{\alpha \xi_t z \mathds{1}_{(L,\infty)}(y)e^{r(T-t)}}Q(dz)f_k(y)dydt \\
&\quad - \lambda f\big(t,X_t^{\xi,b}\big)\sum_{k=1}^m p_k(t) \int_0^\infty h(t,J(p_t,y))f_k(y)dydt + \lambda f\big(t,X_t^{\xi,b}\big)h(t,p_t)dt.
\end{align*}
  Substituting this into \eqref{G}, we obtain
\begin{align*}
&dG\big(t,X^{\xi,b}_t,p_t\big) \\
&=f\big(t,X^{\xi,b}_t\big)\bigg(- \alpha\, e^{r(T-t)}h(t,p_t)\Big((\mu-r)\xi_t+c(b_t)-\frac12\alpha\sigma^2 e^{r(T-t)}\xi_t^2\Big) \\
& + \lambda f\big(t,X^{\xi,b}_t\big)\sum_{k=1}^m p_k(t) \int_0^\infty h(t,J(p_t,y))e^{\alpha b_t y e^{r(T-t)}}\int_{(0,1)}\!e^{\alpha \xi_t z \mathds{1}_{(L,\infty)}(y)e^{r(T-t)}}Q(dz)f_k(y)dy\\
& - \lambda\,h(t,p_t) + h_t(t,p_t)\bigg) dt - f\big(t,X^{\xi,b}_{t-}\big)\,h(t,p_{t-})\alpha\sigma e^{r(T-t)} \xi_t dW_t - f\big(t,X^{\xi,b}_{t-}\big)h(t,p_{t-})\hat\Psi(dt,E)\\
& + \int_E f\big(t,X^{\xi,b}_{t-}\big)\big(h(t,J(p_{t-},y))-h(t,p_{t-})\big)e^{\alpha b_t y e^{r(T-t)}}e^{\alpha \xi_t z \mathds{1}_{(L,\infty)}(y) e^{r(T-t)}}\hat{\Psi}(dt, d(y,z)),
\end{align*}
Therefore, by definition of the operator $\mathcal{H}$ given in~\eqref{H}, we have
\begin{equation*}
d G\big(t,X^{\xi,b}_t,p_t\big) = f\big(t,X^{\xi,b}_{t}\big)\mathcal{H} h(t,p_t;\xi_t,b_t)dt + d\eta^{\xi,b}_t,
\end{equation*}
where $\eta^{\xi,b}_t := \bar\eta^{\xi,b}_t - \hat\eta^{\xi,b}_t - \tilde\eta^{\xi,b}_t$
with
\begin{align*}
\bar\eta^{\xi,b}_t &:= \int_0^t \int_E f\big(s,X^{\xi,b}_{s-}\big)\big(h(s,J(p_{s-},y))-h(s,p_{s-})\big)e^{\alpha b_s y e^{r(T-s)}}e^{\alpha \xi_s z \mathds{1}_{(L,\infty)}(y) e^{r(T-s)}}\hat{\Psi}(dt, d(y,z)), \\
\hat\eta^{\xi,b}_t &:=  \int_0^t f\big(s,X^{\xi,b}_{s-}\big)h(s,p_{s-})\hat\Psi(ds,E), \\
\tilde\eta^{\xi,b}_t &:= \int_0^t f\big(s,X^{\xi,b}_{s-}\big)h(s,p_{s-})\alpha \sigma e^{r(T-s)} \xi_s dW_s.
\end{align*}
To complete the proof we need to show that the introduced processes are martingales w.r.t.\ ${\mathfrak G}$ on $[0,T]$. 
According to Corollary VIII.C4 in \cite{bre}, the process $(\tilde\eta^{\xi,b}_t)_{t\ge0}$ is a martingale w.r.t.\ $\mathfrak{G}$ if 
\begin{equation*}
\Eop\bigg[\int_0^t\!\! \int_E\!\Big| f\big(s,X^{\xi,b}_{s}\big)\big(h(s,J(p_{s},y))-h(s,p_{s})\big)e^{\alpha b_s y e^{r(T-s)}}e^{\alpha \xi_s z \mathds{1}_{(L,\infty)}(y) e^{r(T-s)}}\Big| \hat\nu(ds,d(y,z))\bigg] <\infty.
\end{equation*}
Using the boundedness of $h$ with constant $K_0$, we obtain that the expectation above is less or equal to
\begin{equation*}
\lambda 2K_0 M_Z\big(\alpha K e^{|r|T}\big)\sum_{k=1}^m M_k\big(\alpha e^{|r|T}\big)\int_0^t \Eop\big[\big|f\big(s,X^{\xi,b}_{s}\big)\big|\big]ds,
\end{equation*}
where, by Lemma~\ref{fbounded},
\begin{equation*}
\Eop\big[\big|f\big(s,X^{\xi,b}_{s}\big)\big|\big] = \Eop_{\Q_s^{\xi,b}}\bigg[\frac{\big|f\big(s,X^{\xi,b}_{s}\big)\big|}{L_s^{\xi,b}}\bigg] \le K_2,
\end{equation*}
which yields the desired finiteness.
Similarly the martingale property of $(\hat\eta^{\xi,b}_t)_{t\ge0}$ can be seen.
Moreover, by the boundedness of $h$ and $\xi$ as well as Lemma~\ref{fbounded}, it follows
$$\Eop\Big[\big(f\big(s,X^{\xi,b}_{s-}\big)h(s,p_{s-})\alpha \sigma e^{r(T-s)} \xi_s\big)^2\Big]<\infty,$$
which implies the martingale property of $(\tilde\eta^{\xi,b}_t)_{t\ge0}$.
\end{proof}

The following result can be found in \cite{Mitrinovic1993}.

\begin{lemma}\label{le:ineqsum}
	Let $\alpha_1\le \ldots \le \alpha_n$ and $\beta_1\le\ldots\le\beta_n$ be real numbers and $(p_1,\ldots,p_n)\in\Delta_n$. Then
	\begin{equation*}
	\sum_{j=1}^n p_j\alpha_j\beta_j \ge \sum_{j=1}^n p_j\alpha_j\sum_{k=1}^n p_k\beta_k.
	\end{equation*}
\end{lemma}

\subsection{Proofs}

Recall the function $f:[0,T]\times\R\to\R$ defined by \eqref{f} and the operator $\mathcal{H}$ given by~\eqref{H}.

\begin{proof}[Proof of Theorem~\ref{veri}]
   Let $h:[0,T]\times\Delta_m\to(0,\infty)$ be a function satisfying the conditions stated in the theorem. Note that every Lipschitz function is also absolutely continuous. We set
   \begin{equation*}
   G(t,x,p) := f(t,x)\,h(t,p), \quad (t,x,p)\in[0,T]\times\R\times\Delta_m.
   \end{equation*}
   Let us fix $t\in[0,T]$ and $(\xi,b)\in \mathcal U[t,T]$.
   From Lemma~\ref{characG}, it follows 
   \begin{equation}\label{GT}
   G(T,X^{\xi,b}_T,p_T) = G(t,X^{\xi,b}_t,p_t) + \int_t^T f(s,X^{\xi,b}_s)\mathcal{H} h(s,p_s;\xi_s,b_s)ds + \eta^{\xi,b}_T - \eta^{\xi,b}_t ,
   \end{equation}
   where $\stprT{\eta^{\xi,b}}$ is a martingale w.r.t.\ ${\mathfrak G}$ and we set $\mathcal{H} h(s,p_s;\xi,b)$ to zero at those points $s\in[t,T]$ where $h_t$ does not exist. Note that $h$ is partially differentiable w.r.t.\ $t$ almost everywhere in the sense of the Lebesgue measure according to the absolute continuity of $t\mapsto h(t,p)$ for all $p\in\Delta_m$.
   The generalized HJB equation~\eqref{HHJB} implies 
   \begin{equation*}
   \mathcal{H} h(s,p_s;\xi_s,b_s)\ge0\quad\stT.  
   \end{equation*}
   As a consequence
   \begin{equation*}
   \int_t^T f(s,X^{\xi,b}_s)\,\mathcal{H} h(s,p_s;\xi_s,b_s) ds\le 0,
   \end{equation*}
   due to the negativity of  $f$. Thus, by~\eqref{GT}, we get
   \begin{equation}\label{GGeta}
   G(T,X^{\xi,b}_T,p_T)\le G(t,X^{\xi,b}_t,p_t) + \eta^{\xi,b}_T-\eta^{\xi,b}_t.
   \end{equation}
   Using the boundary condition~\eqref{hHJBbcond}, we obtain
   \begin{equation*}
   G(T,x,p) = f(T,x)h(T,p) = f(T,x) = -e^{-\alpha x}=U(x).
   \end{equation*}
   Now, we take the conditional expectation in \eqref{GGeta} given $(X^{\xi,b}_t,p_t)=(x,p)$ on both sides of the inequality, which yields
   \begin{equation*}
   \Eop^{t,x,p}\big[U(X^{\xi,b}_T)\big]\le G(t,x,p).
   \end{equation*}
   Taking the supremum over all investment and reinsurance strategies $(\xi,b)\in \mathcal U[t,T]$, we obtain 
   \begin{equation}\label{VG}
   V(t,x,p) \le G(t,x,p).
   \end{equation} 
   To show equality, note that  $(\xi^\star_s,b^\star_s)$ given by~\eqref{optstr} (with $g$ replaced by $h$ in $A(s,p)$ and $B(s,p)$) are the unique minimizer of the HJB equation \eqref{HJBg}. 
   Therefore, 
   \begin{equation*}
   \mathcal{L} h(s,p_s;\xi^{\star}_s,b^\star_s) + \inf_{\varphi\in\partial^C h_p(t)}\{\varphi\} = 0.
   \end{equation*}
   So we can deduce that
   \begin{equation*}
   \mathcal{H} h(s,p_{s};\xi^{\star}_s,b^{\star}_s) = 0,\quad s\in[t,T].
   \end{equation*}
   This implies 
   \begin{equation*}
   \int_t^T f(s,X^{\xi^\star,b^\star}_s)\,\mathcal{H} h(s,p_s;\xi^{\star}_s,b^{\star}_s)ds = 0.
   \end{equation*}
   Consequently,
   \begin{equation*}
   U(X^{\xi^\star,b^\star}_T)= G(T,X^{\xi^\star,b^\star}_T,p_T) = G(t,X^{\xi^\star,b^\star}_t,p_t) + \eta^{\xi^\star,b^\star}_T - \eta^{\xi^\star,b^\star}_t.
   \end{equation*}
   Again, taking the conditional expectation given $(X^{\xi^\star,b^\star}_t,p_t)=(x,p)$ on both sides then yields
   \begin{equation*}
   \Eop^{t,x,p}\big[U(X^{\xi^\star,b^\star}_T)\big] = G(t,x,p)= -e^{-\alpha x e^{r(T-t)}}h(t,p)
   \end{equation*}
   and the proof is complete.
\end{proof} 

\begin{proof}[Proof of Lemma \ref{propg}]
	\begin{enumerate}
		\item The boundedness and positivity is proven by the same line of arguments as in \cite[Lemma 4.4 (a)]{BaeuerleLeimcke2020}.
		\item Follows by conditioning. 
		\item Follows again by conditioning. 
      \item The concavity is proven in much the same way as in \cite[Lemma 4.4 (c)]{BaeuerleLeimcke2020}.
		\item The Lipschitz condition is proven in much the same way as in \cite[Lemma 6.1 (d)]{BaeuerleRieder2007}.\qedhere
	\end{enumerate}
\end{proof}

\begin{proof}[Proof of Theorem~\ref{existenceHJB}]
   Fix $t\in[0,T)$ and $\xib\in \mathcal U[t,T]$.
   Let $\tau$ be the first jump time of $X^{\xi,b}$ after $t$ and $t'\in(t,T]$. It follows from Lemma~\ref{propg} and Lemma \ref{characG} that 
   \begin{equation}\label{eqprDir:Vtilde}
   V(\tau\wedge t',X^{\xi,b}_{\tau\wedge t'},p_{\tau\wedge t'}) 
   = V(t,X^{\xi,b}_t,p_t) + \int_t^{\tau\wedge t'} f(s,X^{\xi,b}_s)\,\mathcal{H}g(s,p_s;\xi_s,b_s)ds + \eta^{\xi,b}_{\tau\wedge t'} - \eta^{\xi,b}_t,
   \end{equation}
   where $\stprT{\eta^{\xi,b}}$ is a martingale w.r.t.\ ${\mathfrak G}$ and we set $\mathcal{H}g(s,p_s;\xi_s,b_s)$ to zero at those $s\in[t,T]$ where $g_t(s,p_s)$ does not exist. 
   For any $\varepsilon>0$ we can construct  a strategy $(\xi^\varepsilon,b^\varepsilon)\in \mathcal U[t,T]$ with $(\xi^\varepsilon_s,b^\varepsilon_s)=(\xi_s,b_s)$ for all $s\in[t,\tau\wedge t']$  from the continuity of $V$ such that
   \begin{align*}
   \Eop^{t,x,p}\Big[V(\tau\wedge t',X^{\xi,b}_{\tau\wedge t'},p_{\tau\wedge t'})\Big]
   &\le \Eop^{t,x,p}\Big[\Eop^{\tau\wedge t',X^{\xi,b}_{\tau\wedge t'},p_{\tau\wedge t'}}\Big[U(X_T^{\xi^\varepsilon,b^\varepsilon})\Big]\Big] +\varepsilon 
   \le \Eop^{t,x,p}\Big[U(X_T^{\xi^\varepsilon,b^\varepsilon})\Big]+\varepsilon \\
   &\le V(t,x,p)+\varepsilon.
   \end{align*}
   From the arbitrariness of $\varepsilon>0$ we conclude
   $$V(t,x,p) \ge \Eop^{t,x,p}\Big[V(\tau\wedge t',X^{\xi,b}_{\tau\wedge t'},p_{\tau\wedge t'})\Big]. $$
   Using this statement and \eqref{eqprDir:Vtilde} we obtain 
   \begin{align*}
   0 &\ge \lim_{t'\downarrow t}\Eop^{t,x,p}\bigg[\frac{1}{t'-t}\int_t^{t'} f(s,X^{\xi,b}_s)\,\mathcal{H}g(s,p_s;\xi_s,b_s) ds \big| t'<\tau\bigg]\Pop^{t,x,p}(t'<\tau) \\
   &\quad + \lim_{t'\downarrow t}\Eop^{t,x,p}\bigg[\frac{1}{t'-t}\int_t^{\tau} f(s,X^{\xi,b}_s)\,\mathcal{H}g(s,p_s;\xi_s,b_s)ds\big| t'\ge\tau\bigg]\Pop^{t,x,p}(t'\ge\tau),
   \end{align*}
  where
   \begin{equation*}
   \lim_{t'\downarrow t}\Pop^{t,x,p}(\tau\le t') = 1-\lim_{t'\downarrow t}e^{-\lambda(t'-t)}=0.
   \end{equation*}
   Consequently,
   \begin{equation*}
   0\ge \lim_{t'\downarrow t}\Eop^{t,x,p}\bigg[\frac{1}{t'-t}\int_t^{t'} f(s,X^{\xi,b}_s)\mathcal{H}g(s,p_s;\xi_s,b_s)ds\Ind{t'<\tau}\bigg].
   \end{equation*}
   By the dominated convergence theorem, we can interchange the limit and the expectation and we obtain by the fundamental theorem of Lebesgue calculus and $\Ind{t'<\tau}\to1$ $\Pop$-a.s.\ for $t'\downarrow t$, 
   \begin{equation*}
   0\ge \Eop^{t,x,p}\bigg[f(t,X^{\xi,b}_t)\,\mathcal{H}g(t,p_t;\xi_t,b_t)\bigg].
   \end{equation*}
   From now on, let $\xib\in[-K,K]\times[0,1]$ and $\varepsilon>0$ as well as $(\bar{\xi},\bar{b})\in \mathcal U[t,T]$ be a fixed strategy with $(\bar{\xi}_s,\bar{b}_s)\equiv(\xi,b)$ for $s\in[t,t+\varepsilon)$. Then
   \begin{equation*}
   0\ge \Eop^{t,x,p}\bigg[f(t,X^{\bar{\xi},\bar{b}}_t)\,\mathcal{H}g(t,p_t;\bar{\xi}_t,\bar{b}_t)\bigg] = f(t,x)\mathcal{H}g(t,p;\xi,b)
   \end{equation*}
   at those points $(t,p)$ where $g_t(t,p)$ exists. Due to the negativity of $f$, we get
   \begin{equation*}
   0\le \mathcal{H}g(t,p;\xi,b).
   \end{equation*}
   We show next the inequality above if $g_t$ does not exist. For this purpose, we denote by $M_p\subset[0,T]$ the set of points at which $g_p^\prime(t)$ exists for any $p\in\Delta_m$. On the basis of Theorem~\ref{th:genCgco}, we have, for any $p\in\Delta_m$, 
   \begin{equation*}
   \partial^C g_p(t) = co\Big\{\lim_{n\to\infty} g_p^\prime(t_n): t_n\to t, t_n\in M_p\Big\}.
   \end{equation*}
   That is, for every $\varphi\in\partial^C g_p(t)\subset[0,T]$, there exists $u\in\N$ and $(\beta_1,\ldots,\beta_u)\in\Delta_u$ such that $\varphi = \sum_{i=1}^u \beta_i\,\varphi^i$, where $\varphi^i = \lim_{n\to\infty} g_p(t_n^i)$ for sequences $(t_n^i)_{n\in\N}$ with $\lim_{n\to\infty}t_n^i=t$ along existing $g_p^\prime$. From what has already been proved, it can be concluded that, for any $i=1,\ldots,u$
   \begin{equation*}
   0 \le \mathcal{L}g(t_n^i,p;\xi,b)+g_t(t_n^i,p).
   \end{equation*}
   Thus, by the continuity of $t\mapsto g(t,p)$, $p\mapsto g(t,p)$ and $p\mapsto J(p,y)$, we get for $i=1,\ldots, u$
   \begin{equation*}
   0 \le \beta_i\mathcal{L}g(t,p;\xi,b)+\beta_i\lim_{n\to\infty}g_t(t_n^i,p), 
   \end{equation*}
   which yields
   \begin{equation*}
   0 \le \mathcal{L}g(t,p;\xi,b)+\sum_{i=1}^u\beta_i\lim_{n\to\infty}g_t(t_n^i,p)
   = \mathcal{L}g(t,p;\xi,b)+ \varphi.
   \end{equation*}
   Due to the arbitrariness of $\varphi\in\partial^C g_p(t)$ and $(\xi,b)\in[-K,K]\times[0,1]$, we obtain 
   \begin{equation*}
   0 \le \inf_{(\xi,b)\in[-K,K]\times[0,1]}\mathcal{L}g(t,p;\xi,b)+ \inf_{\varphi\in\partial^C g_p(t)}\{\varphi\}.
   \end{equation*}
   Our next objective is to establish the reverse inequality. For any $\varepsilon>0$ and $0\le t<t'\le T$, there exists a strategy $({\xi^{\varepsilon,t^\prime},b^{\varepsilon,t^\prime}})\in \mathcal{U}[t,T]$ such that
   \begin{equation*}
   V(t,x,p)-\varepsilon(t'-t)\le \Eop^{t,x,p}\Big[U\big(X_T^{\xi^{\varepsilon,t^\prime},b^{\varepsilon,t^\prime}}\big)\Big] 
   \le \Eop^{t,x,p}\Big[V\big(\tau\wedge t',X^{\xi^{\varepsilon,t^\prime},b^{\varepsilon,t^\prime}}_{\tau\wedge t'},p_{\tau\wedge t'}\big)\Big].
   \end{equation*}
   Using Lemma \ref{characG} it follows
   \begin{equation*}
   -\varepsilon(t'-t)\le \Eop^{t,x,p}\bigg[\int_t^{\tau\wedge t'} f\big(s,X^{\xi^{\varepsilon,t^\prime},b^{\varepsilon,t^\prime}}_s\big)\,\mathcal{H}g\big(s,p_s;\xi^{\varepsilon,t^\prime}_s,b^{\varepsilon,t^\prime}_s\big)ds\bigg].
   \end{equation*}
   In the same way as before, we get 
   \begin{align*} 
   -\varepsilon&\le \lim_{t'\downarrow t}\Eop^{t,x,p}\bigg[\frac{1}{t'-t}\int_t^{t'} f\big(s,X^{\xi^{\varepsilon,t^\prime},b^{\varepsilon,t^\prime}}_s\big)\,\mathcal{H}g\big(s,p_s;\xi^{\varepsilon,t^\prime}_s,b^{\varepsilon,t^\prime}_s\big)ds\Ind{t'<\tau}\bigg] \\
   &\le\lim_{t'\downarrow t}\Eop^{t,x,p}\bigg[\frac{1}{t'-t}\int_t^{t'}\! f\big(s,X^{\xi^{\varepsilon,t^\prime},b^{\varepsilon,t^\prime}}_s\big)\!\inf_{(\xi,b)\in[-K,K]\times[0,1]}\mathcal{H}g\big(s,p_s;\xi,b\big)ds\Ind{t'<\tau}\bigg].
   \end{align*}
   We can again interchange the limit and the infimum by the dominated convergence theorem which yields 
   \begin{equation*}
   -\varepsilon\le \Eop^{t,x,p}\bigg[\lim_{t'\downarrow t}\frac{1}{t'-t}\int_t^{t'} f\big(s,X^{\xi^{\varepsilon,t^\prime},b^{\varepsilon,t^\prime}}_s\big)\,\inf_{(\xi,b)\in[-K,K]\times[0,1]}\mathcal{H}g\big(s,p_s;\xi,b\big)ds\Ind{t'<\tau}\bigg].
   \end{equation*}
   Thus the same conclusion can be draw as above, i.e. 
   \begin{equation*}
   -\varepsilon \le f(t,x)\inf_{(\xi,b)\in[-K,K]\times[0,1]}\mathcal{H}g(t,p;\xi,b)
   \end{equation*}
   at those point where $g_t(s,p)$ exists. According to the negativity of $f$ and the arbitrariness of $\varepsilon>0$, we get, by $\varepsilon\downarrow0$,
   \begin{equation*}
   0 \ge \inf_{(\xi,b)\in[-K,K]\times[0,1]}\mathcal{H}g(t,p;\xi,b)
   \end{equation*}
   at those point where $g_t(s,p)$ exists.
   By the same way as before, we obtain in the case of no differentiability of $g$ w.r.t.\ $t$, that 
   \begin{equation*}
   0 \ge \inf_{(\xi,b)\in[-K,K]\times[0,1]}\mathcal{L}g(t,p;\xi,b) + \inf_{\varphi\in\partial^C g_p(t)}\{\varphi\}.
   \end{equation*}
   Summarizing, we have equality in the previous expression.
   The optimality of $(\xi^{\star},b^{\star})$  follows as in the proof of Theorem \ref{veri}.
\end{proof}





\begin{thebibliography}{99}

\bibitem{acw11} B. Avanzi, L.C. Cassar and B. Wong. Modelling dependence in insurance claims processes with L\'evy copulas. UNSW Australian School of Business Research Paper, (2011ACTL01), 2011.

\bibitem{atwy14} B. Avanzi, J. Tao, B. Wong and X. Yang. Capturing non-exchangeable dependence in multivariate loss processes with nested Archimedean L\'evy copulas. Annals of Actuarial Science 10, 87-117, 2016.

\bibitem{AsmussenAlbrecher2010} H. Albrecher and S. Asmussen. Ruin probabilities. World Scientific, Singapore, 2010.


\bibitem{bb11} N. B\"auerle and A. Blatter. Optimal control and dependence modeling of insurance portfolios with L\'{e}vy dynamics. Insurance: Mathematics and Economics 48(3), 398--405, 2011.

\bibitem{BaeuerleGruebel2005} N. B\"auerle and R. Gr\"ubel. Multivariate counting processes: Copulas and beyond. Advances in Applied Probability 35(2), 379--408, 2005.

\bibitem{BaeuerleGruebel2008} N. B\"auerle and R. Gr\"ubel. Multivariate  risk  processes  with  interacting  intensities. Advances in Applied Probability 40(2), 578--601, 2008.

\bibitem{BaeuerleLeimcke2020} N. B\"auerle and G. Leimcke. Robust optimal investment and reinsurance problems with learning. Scandinavian Actuarial Journal, DOI: 10.1080/03461238.2020.1806917, 2020.


\bibitem{BaeuerleRieder2007} N. B\"auerle and U. Rieder. Portfolio optimization with jumps and unobservable intensity. Mathematical Finance 17(2), 205--224, 2007.



\bibitem{BiChen2019} J. Bi and K. Chen. Optimal investment-reinsurance problems with common shock dependent risks under two kinds of premium principles. RAIRO Operations Research 53(1), 179--206, 2019.

\bibitem{BLX16} J. Bi, Z. Liang and F. Xu. Optimal mean–-variance investment and reinsurance problems for the risk model with common shock dependence. Insurance: Mathematics and Economics, 70, 245-258, 2016.

\bibitem{bs20} M. Brachetta and H. Schmidli, H. Optimal reinsurance and investment in a diffusion model. Decisions in Economics and Finance 43, 341–-361 (2020).


\bibitem{bk05} Y. Bregman and C. Kl\"uppelberg, C. (2005). Ruin estimation in multivariate models with Clayton dependence structure. Scandinavian Actuarial Journal, 2005(6), 462-480.






\bibitem{bre} P. Br\'emaud. Point processes and queues. Springer-Verlag, New York, 1981.



\bibitem{Clarke1983} F.H. Clarke. Optimization and nonsmooth analysis. Canadian Mathematical Society Series of Monographs and Advanced Texts - A Wiley-Interscience Publication, New York, 1983.




\bibitem{GBC12} L. Gong, A.L. Badescu and E.C. Cheung. Recursive methods for a multi-dimensional risk process with common shocks. Insurance: Mathematics and Economics, 50(1), 109-120, 2012.







\bibitem{gl20} G. Leimcke. Bayesian optimal investment and reinsurance to maximize exponential utility of terminal wealth for an insurer with various lines of business. PhD Thesis, Karlsruhe Institute of Technology, 2020.

\bibitem{LST14} G. Leobacher, M. Sz\"olgyenyi and S. Thonhauser. Bayesian dividend optimization and finite time ruin probabilities. Stochastic Models, 30(2), 216-249, 2014.

\bibitem{LiangBayraktar2014} Z. Liang and E. Bayraktar. Optimal reinsurance and investment with unobservable claim size and intensity. Insurance: Mathematics and Economics 55, 156--166, 2014.



\bibitem{Mitrinovic1993} D.S. Mitrinovic, J. Pecaric and A.M. Fink. Classical and new inequalities in analysis. Mathematics and its Applications. Kluwer Academic Publishers, Dordrecht,  1993.

\bibitem{MuellerStoyan2002} A. M\"{u}ller and D. Stoyan. Comparison for stochastic models and risks. Wiley, New York, 2002.

\bibitem{Protter2005} P. Protter. Stochastic Integration and Differential Equations. Springer, Berlin, 2nd edition, 2005.

\bibitem{SchererSelch2018} M. Scherer and  D. Selch. A multivariate claim count model for applications in insurance. Springer Actuarial. Springer International Publishing, Cham, 2018.



\bibitem{s15} M. Sz\"olgyenyi, Dividend maximization in a hidden Markov switching model. Statistics \& Risk Modeling, 32(3-4), 143-158, 2015.

\bibitem{Wang2018} K. Wang, M. Gao, Y. Yang and Y. Chen. Asymptotics for the finite-time ruin probability in a discrete-time risk model with dependent insurance and financial risks. Lithuanian Mathematical Journal, 58(1), 113--125, 2018.


\bibitem{YuenLiangZhou2015} K.C. Yuen, Z. Liang abd M.Zhou. Optimal proportional reinsurance with common shock dependence. Insurance: Mathematics and Economics 64, 1--13, 2015.



\bibitem{ZhuShi2019}
S. Zhu and J. Shi. Optimal reinsurance and investment strategies under mean-variance criteria: Partial and full information. arXiv e-print 1906.08410v3, 2020.



\end{thebibliography}


\end{document}